\documentclass[conference]{IEEEtran}

\usepackage{amsmath}
\usepackage{amsfonts}
\usepackage{amssymb}
\usepackage{amsthm}
\usepackage{xcolor}
\usepackage{array}
\PassOptionsToPackage{hyphens}{url}
\usepackage[hidelinks]{hyperref}

\usepackage{tikz}
\usetikzlibrary{
  arrows.meta,positioning,calc,
  shapes.misc,backgrounds,fit
}
\usepackage{subcaption}
\usepackage{stfloats}
\usepackage{arydshln}

\usepackage{paralist}
\usepackage{relsize}
\usepackage{booktabs}
\usepackage{xstring}
\usepackage{seqsplit}
\usepackage{float}
\newfloat{algorithm}{tbp}{loa}
\floatname{algorithm}{Algorithm}
\newcommand{\addr}[1]{%
  \texttt{\StrLeft{#1}{6}\ldots\StrRight{#1}{4}}%
}
\definecolor{okBlue}{HTML}{0072B2}      % the position being unwound
\definecolor{okOrange}{HTML}{E69F00}    % rewards being sold
\definecolor{okVermillion}{HTML}{D55E00}% the asset that appears to loop
\definecolor{okGreen}{HTML}{009E73}     % where the value ends up
\definecolor{okGrey}{HTML}{555555}      % structure, no semantics
\definecolor{okSkyBlue}{HTML}{56B4E9}
\definecolor{okPurple}{HTML}{CC79A7}
\definecolor{darkgreen}{HTML}{009E73}
\definecolor{darkred}{HTML}{D55E00}
\definecolor{poolblue}{HTML}{0072B2}
\definecolor{acctgray}{RGB}{180,180,180}
\definecolor{scamber}{HTML}{E69F00}

\newtheorem{definition}{Definition}
\newtheorem{property}{Property}
\newtheorem{theorem}{Theorem}

\newtheorem{proposition}{Proposition}

\title{If It Walks Like an Arbitrage: Protocol-Agnostic Detection with Decidable Structural Equivalence}

\author{%
\IEEEauthorblockN{Adam Khayam, Hamid Kolli,
  Mohamed Iguernlala, \c{C}agdas Bozman}
\IEEEauthorblockA{Functori, Paris, France}}

\IEEEoverridecommandlockouts
\makeatletter\def\@IEEEpubidpullup{6.5\baselineskip}\makeatother
\IEEEpubid{\parbox{\columnwidth}{
    Preprint.  Under submission to a peer-reviewed
    security venue.
}
\hspace{\columnsep}\makebox[\columnwidth]{}}

\begin{document}

\maketitle

\IEEEpeerreviewmaketitle

\begin{abstract}
 Whether a transaction performed an arbitrage, and
  by which route, is a question asked of its
  execution trace after the fact.  We conjecture
  that such traces admit a normal form on which
  questions of this kind become queries, and we
  test it by building one and putting it to work.
  Each trace becomes an abstract syntax tree of
  token transfers, grouped by call-frame nesting.
  A term rewriting system of 16~rules reduces it.  Rewriting terminates and
  carries exactly the transfers of the trace,
  whichever order the rules fire in.  Under a
  deterministic kernel scanning the EVM-fixed trace
  order, every trace has a unique normal form, and
  the structural equivalence this induces on fund
  flows is decidable.  Preservation, termination, soundness,
  uniqueness and decidability are mechanized in
  Rocq with zero admitted obligations.  We report
  arbitrage detection in full: cycles emerge at the
  fixpoint and are read off the normal form with no
  protocol-specific patterns.  Detection is the query
  we evaluate at scale; structural equivalence is a
  second query.  The pipeline depends only
  on the standard ERC token and WETH ABIs, so the
  same binary runs unmodified on Arbitrum and
  BSC\@.

  We evaluate on two arbitrary block ranges,
  analysed in full with no transaction excluded:
  220\,000 Ethereum blocks against EigenPhi, a
  widely used MEV detection platform, and 1\,000
  shared blocks against ArbiNet, a graph neural
  network classifier.  These are the only tools and
  label sets in this domain we were able to reuse. 
  We report 469\,801 confirmed arbitrages,
  overlapping 83.5\% with EigenPhi and 81\% with
  ArbiNet, together with 245\,497 attempted
  arbitrages and 60\,199 confirmed detections
  EigenPhi does not report.  99.2\% follow from the
  fixpoint alone and are sound over the decoded
  transfers, and
  manual validation of 500~transactions finds no
  false positive among the confirmed.  We also
  analyse baseline-only detections to characterize
  where the systems disagree.
\end{abstract}

\section{Introduction}\label{sec:intro}

% ── P1: The real problem ──
Every Ethereum transaction leaves a replayable trace of
its internal calls and token movements.  The record is
complete, yet not self-explanatory: the same movements
can be a trade, a repayment, a withdrawal, or an
attack, and which one happened depends on how the
execution produced the transfers, not on the transfers
alone.  The difference now carries legal weight.
Whether a twelve-second interaction with the
block-building pipeline was arbitrage or wire fraud is
the substance of a federal prosecution whose first jury
could not decide~\cite{doj2024peraire}; a fraud
conviction over the Mango Markets exploit, defended as
a trading strategy, was vacated on the ground that an
automated contract cannot be deceived, and is on
appeal~\cite{cftc2023eisenberg}; and MiCA now requires
venues to report manipulative
transaction-ordering~\cite{esma2025mica}, without
saying how.  Surveillance, protocol design, and
incident response ask the same questions of a trace
after the fact: did it perform an arbitrage?  Through
which route?  Do two transactions, under different
contracts and tokens, implement the same strategy?

% ── Figure 1: three arbitrage topologies ──
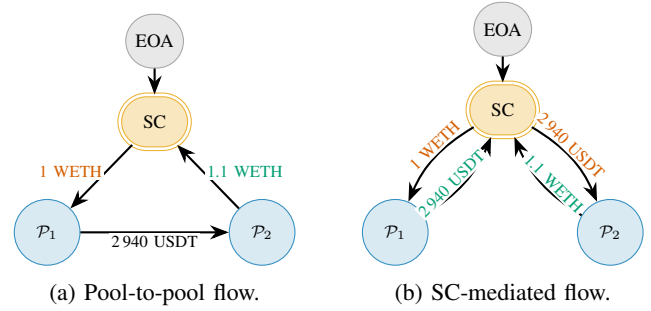
\begin{figure}[t]
\tikzset{
    acct/.style={rounded rectangle,
                 draw=gray!70, fill=acctgray!30,
                 minimum width=1.2cm,
                 minimum height=0.9cm,
                 text=black, font=\small,
                 align=center},
    sc/.style={rounded rectangle,
               draw=scamber!80, fill=scamber!25,
               minimum width=1.4cm,
               minimum height=0.9cm,
               text=black, font=\small,
               align=center,
               double, double distance=1pt},
    pool/.style={circle, draw=poolblue!80,
                 fill=poolblue!15,
                 minimum size=1.1cm, text=black,
                 font=\footnotesize, align=center},
    xfer/.style={-{Stealth[length=2.5mm]}, thick},
    elbl/.style={font=\footnotesize, fill=white,
                 inner sep=1pt},
}
\centering
\begin{subfigure}[t]{0.48\columnwidth}
\centering
\begin{tikzpicture}[font=\footnotesize, 
  scale=0.79, every node/.style={scale=0.79}]
\node[acct] (eoa) at (0,3.7)    {EOA};
\node[sc]   (sc)  at (0,2.35)   {SC};
\node[pool] (p1)  at (-1.8,0.5) {$\mathcal{P}_1$};
\node[pool] (p2)  at (1.8,0.5)  {$\mathcal{P}_2$};
\draw[xfer] (eoa) -- (sc);
\draw[xfer] (sc) --
  node[elbl,above left=-2pt]
    {\textcolor{darkred}{1 WETH}} (p1);
\draw[xfer] (p1) --
  node[elbl,below]
    {2\,940 USDT} (p2);
\draw[xfer] (p2) --
  node[elbl,above right=-2pt]
    {\textcolor{darkgreen}{1.1 WETH}} (sc);
\end{tikzpicture}
\caption{Pool-to-pool flow.}
\label{fig:arbi:a}
\end{subfigure}\hfill
\begin{subfigure}[t]{0.48\columnwidth}
\centering
\begin{tikzpicture}[font=\footnotesize, 
  scale=0.79, every node/.style={scale=0.79}]
\node[acct] (eoa) at (0,3.7)    {EOA};
\node[sc]   (sc)  at (0,2.35)   {SC};
\node[pool] (p1)  at (-1.8,0.3) {$\mathcal{P}_1$};
\node[pool] (p2)  at (1.8,0.3)  {$\mathcal{P}_2$};
\draw[xfer] (eoa) -- (sc);
\draw[xfer] (sc) to[bend right=18]
  node[elbl,sloped,above,pos=0.4]
    {\textcolor{darkred}{1 WETH}} (p1);
\draw[xfer] (p1) to[bend right=18]
  node[elbl,sloped,above,pos=0.4]
    {\textcolor{darkgreen}
      {2\,940 USDT}} (sc);
\draw[xfer] (sc) to[bend left=18]
  node[elbl,sloped,above,pos=0.4]
    {\textcolor{darkred}
      {2\,940 USDT}} (p2);
\draw[xfer] (p2) to[bend left=18]
  node[elbl,sloped,above,pos=0.4]
    {\textcolor{darkgreen}
      {1.1 WETH}} (sc);
\end{tikzpicture}
\caption{SC-mediated flow.}
\label{fig:arbi:b}
\end{subfigure}

\caption{One arbitrage, two transfer graphs.
  \textcolor{darkred}{Red}: sent;
  \textcolor{darkgreen}{green}: received.}
\label{fig:arbidef}
\end{figure}

% ── P2: What an arbitrage is, and how it is exploited ──
An \textit{arbitrage}, informally, is a round trip that
pays: value leaves an account in some asset and returns
to it, in the same asset, in greater quantity.  On
Ethereum it is an industry: prices for a pair differ
across automated market makers
(AMM)~\cite{uniswap,curve,balancer}, and bots
buy where an asset is cheap and sell where it is dear,
closing the loop in one atomic transaction so the trade
either profits or costs only gas.  This is the core of
\textit{maximal extractable value}
(\textbf{MEV}), and the same mechanics, amplified by
flash
loans~\cite{qin2021flashloan,bzx,palkeo2020bzx,flashloan2026},
have drained millions from DeFi pools.  Both
conditions, returning \emph{to the same account} and
\emph{in the same asset}, are claims about how an
execution unfolded; the transfers alone do not carry
them.

Figure~\ref{fig:arbidef} shows the simplest case
twice.  In each panel, an externally owned account
(EOA) invokes a smart contract (SC) that routes
\textsc{WETH} through pools with different reserve
ratios and recovers more than it sent.  The two panels
are one strategy routed differently: the intermediate
\textsc{USDT} passes directly between pools in~(a) and
through the SC in~(b), two transfers against four.
The strategy is the same and the transfer graphs are
not.  Parallel paths that split and merge go further
still: Figure~\ref{fig:txR} shows a real one,
$\mathit{tx}_R$, the running example of the paper.
These shapes recur across hundreds of pool contracts.

% ── P3: The limits of reading transfers alone ──
Detectors read the transfers and look for the shape.
EigenPhi, the production MEV platform whose labels are
the industry reference,\footnote{EigenPhi's public
service was discontinued in June
2026~\cite{eigenphi_shutdown}; we compare against
labels collected while it operated, which ship with our
artifact (Section~\ref{sec:eval}).} flattens a
transaction into its transfer graph, extracts strongly
connected components, and reports an arbitrage when the
balance nearest the sender is
positive~\cite{eigenphi_method}.  Read this way,
cases~(a) and~(b) are different objects though they are
one strategy, and a withdrawal returning an asset to
its owner is indistinguishable from a trade that
profits (Appendix~\ref{sec:casestudy:tn};
Section~\ref{sec:eval} measures how often).

The ceiling is not specific to one tool.
DeFiRanger~\cite{defirangertdsc}, closest to our work,
matches a hand-maintained catalog of attack patterns
against a \textit{cash flow tree}, with no formal
guarantees.  Graph-theoretic
methods~\cite{mclaughlin2023,zhang2024,wang2022cyclic}
enumerate cycles in aggregate token graphs at scale but
collapse the call hierarchy separating a \textit{swap}
from a \textit{relay}: alike in a flat graph, different
in meaning.  Machine-learning
classifiers~\cite{jin2022,niedermayer2024,arbinet2024}
and richer representations~\cite{qin2025clue} produce
labels without the chains that explain them.  A
survey~\cite{mevsurvey2024} confirms the common crutch,
protocol-specific knowledge, and that detectors degrade
on unknown contracts~\cite{chi2024,li2023actlifter}.
Missing is a \emph{structural account} of what arbitrage
\emph{is} inside a trace.

% ── P4: The insight ──
Our key observation is that an execution trace
admits a normal form, under which structural
equivalence of fund flows is \emph{decidable}.  We
turn each trace into an abstract syntax tree of
token transfers and rewrite it, chaining, merging
and annotating, until fixpoint.  Rewriting invents
nothing and terminates: the result carries exactly
the transfers of the trace, reordered, with every
account's balance unchanged whichever order the
rules fire in.  Our kernel scans the EVM-fixed
execution order left to right, and under it the
trace has a unique normal form.

Arbitrage detection is the demonstrating witness:
cycles are never searched for, they
\emph{emerge}.  If a set of transfers walks like
an arbitrage, the system says
so.\footnote{After the duck test, attributed to
James Whitcomb Riley.}  We claim soundness, not
completeness, and state it of the transaction, not
of our labels: when the system reports an
arbitrage, that transaction's transfer graph
contains a closed walk returning to its origin in
a token-equivalent asset with strictly positive
balance (Definition~\ref{def:arb}), naming neither
the rewriting nor the normal form.  The decoder
uses only the standard ERC token and WETH ABIs, so
designs reusing the \texttt{Transfer} signature need
no protocol-specific decoding.

% ── P5: Contributions ──
We contribute (i)~the \emph{substrate}, a rewriting
system whose normal form makes structural equivalence
decidable; (ii)~a \emph{sender-enriched AST}, in which
each transfer carries a \textit{sender} field~$\sigma$
derived from the call hierarchy that separates a swap
from a relay without event catalogs; (iii)~\emph{Rocq
proofs}, with zero admitted obligations, of
preservation, termination, soundness over the
transaction's own transfer graph, uniqueness and
decidable equivalence; and (iv)~\emph{detection as the
first query}, in which flash-loan arbitrages emerge
structurally and verdicts carry their transfer chains,
profit, and reasons.

% ── P6: Validation hook ──
Section~\ref{sec:eval} evaluates the detector against
EigenPhi and ArbiNet~\cite{arbinet2024}.  The Rocq
sources, evaluation pipeline, and binaries are
available at
\url{https://anonymous.4open.science/r/artifact-submission-2128}.

\section{Preliminaries}\label{sec:scamm}

% CUT: section mini-roadmap (subsection headings do this)
% We begin with the execution model that produces the
% traces our algorithm consumes
% (Section~\ref{sec:bg:exec}), recall how automated
% market makers create arbitrage opportunities
% (Section~\ref{sec:bg:amm}), define the data
% structures used throughout the paper
% (Section~\ref{sec:bg:tg}), and close with a pipeline
% overview (Section~\ref{sec:arch}).

\subsection{Ethereum Execution Model}\label{sec:bg:exec}

Ethereum is a distributed state machine whose state is
advanced by \textit{transactions}.  A transaction is
signed by an \textit{externally owned account}
(\textbf{EOA}), an address controlled by a private
key, and may invoke a \textit{smart contract}
(\textbf{SC}), a program deployed at a distinct address
on the blockchain.  Smart contracts can, in turn, invoke
other contracts via internal calls (\texttt{CALL},
\texttt{DELEGATECALL}, \texttt{STATICCALL}), forming a
nested call tree.

\paragraph{Execution trace}
Replaying a transaction through the Ethereum Virtual
Machine (EVM) tracer produces an \textit{execution
trace}: the ordered, hierarchical record of every
internal call and its side effects.  It preserves the
nesting of the call tree and records, per frame, the
caller, callee, call value, input and return data, and
emitted event logs.

\paragraph{Token transfers}
A \textit{token transfer} is a movement of value from
one address to another, and the trace records three
forms: \textit{native ETH}, as value attached to a
\texttt{CALL} opcode; \textit{ERC-20}, as
\texttt{Transfer(\allowbreak from,\allowbreak
to,\allowbreak amount)} event
logs~\cite{erc20}; and \textit{NFT}, as ERC-721
\texttt{Transfer} or ERC-1155
\texttt{Transfer\-Single}/\texttt{Transfer\-Batch}
events carrying a token identifier rather than a
fungible amount (only fungible amounts contribute to
profit).
We treat them uniformly for cycle construction: a
transfer is a tuple
$t = (s, d, a, \tau, \sigma)$ where
$s$ is the source address, $d$ the destination, $a$
the amount or token identifier, $\tau$ the token type,
and $\sigma$ the \textit{sender}: the address that
initiated the nearest enclosing \texttt{CALL} frame
in which the transfer event was emitted, derived
from the call hierarchy of the execution trace.
We write $s(t)$, $d(t)$, $a(t)$, $\tau(t)$ and
$\sigma(t)$ for the five components of a transfer~$t$.
\texttt{DELEGATE\-CALL} frames do not update~$\sigma$
(delegated code executes in the caller's context),
so proxy contracts are handled correctly:
$\sigma$~reflects the external caller, not the proxy.
$\sigma$ is what separates a swap from a relay: in a
swap both transfers share the external caller that
orchestrated them, whereas in a relay the intermediary
is itself the sender of the onward transfer.  The
transfers look alike; their senders do not.

\paragraph{Transaction costs}
A transaction incurs a burned \textit{base fee}
(EIP-1559), a \textit{priority fee} to the block
builder, and optionally a direct contract-to-coinbase
payment, used by MEV arbitrageurs to secure favorable
ordering.

\paragraph{Automated market makers}\label{sec:bg:amm}
An \textit{automated market maker} (\textbf{AMM}) is a
smart contract holding reserves of two or more tokens in
a \textit{liquidity pool} and pricing them
algorithmically, without an order book.  Each pool
prices from its own reserves, so two pools over the same
pair may offer different rates, creating an
\textit{arbitrage opportunity}: buy cheaply in one, sell
for more in the other.

\subsection{Transfer Graphs and Cash Flow Trees}
\label{sec:bg:tg}

We define the data structures and operators used
in the algorithm and the proofs.

\begin{definition}[Transfer graph]\label{def:tg}
Given a transaction $\mathit{tx}$, its \textit{transfer
graph} $G = (V, E)$ is the directed multigraph whose
vertices $V$ are the addresses involved in $\mathit{tx}$
and whose edges $E$ are the transfers $\mathit{tx}$
performs, one edge per transfer, each edge~$t$ directed
from $s(t)$ to $d(t)$.
\end{definition}

Transfer graphs are flat: they record \textit{what} was
transferred, not the call hierarchy that produced it,
and that hierarchy is what separates a round trip from a
coincidence.  The transaction of
Appendix~\ref{sec:casestudy:tn} makes this concrete.
Its flat graph shows \textsc{WETH} leaving the contract
and \textsc{WETH} coming back, and stops there; the call
hierarchy shows the returning \textsc{WETH} to be the
proceeds of two reward sales in frames of their own,
which is what distinguishes repayment from gain.  To
keep that structure, we introduce the \textit{cash flow
tree}.

\begin{definition}[Cash Flow Tree]\label{def:cft}
A \textit{cash flow tree} (CFT) over a transaction
$\mathit{tx}$ is defined inductively:
\begin{itemize}
  \item a \textit{leaf} $\mathsf{Leaf}(t)$ carries a
    single transfer~$t$;
  \item a \textit{chain node} $\mathsf{Chain}(C)$ carries
    a non-empty sequence $C$ of transfers;
  \item a \textit{call node}
    $\mathsf{Tree}(v, [T_1, \ldots, T_m])$ carries the
    address~$v$ whose code was executing and its children
    $T_1, \ldots, T_m$, themselves CFTs, ordered by their
    position in the execution trace.
\end{itemize}
\end{definition}

The CFT mirrors the EVM call stack: each internal
\texttt{CALL} becomes a $\mathsf{Tree}$ node, each token
transfer a $\mathsf{Leaf}$, and a node's children are
what occurred \textit{within} that frame.  Siblings were
therefore produced by the same contract invocation, and
that is the structure the rewriting exploits.
The tree the decoder builds has leaves and call nodes
only; chain nodes hold the sequences the rewriting
assembles, so they appear once reduction begins.
We call it an \textit{abstract syntax tree}
(\textbf{AST}), since the algorithm rewrites it
syntactically rather than interpreting it as a graph.

The rewriting links sibling transfers into chains, then
asks whether a chain closes on itself and pays.  The
remaining definitions make that precise, beginning with
what a set of transfers does to an account.

\begin{definition}[Balance]\label{def:balance}
Let $t_1, t_2, \ldots, t_k$ be transfers.  Their
\textit{balance} $\delta$ sends an
address~$v$ and a token~$\tau$ to the net signed amount
that $v$ receives in~$\tau$: each $t_i$ carrying exactly
token~$\tau$ credits $d(t_i)$ and debits $s(t_i)$ by
$a(t_i)$, and leaves every other address untouched.  We
write $\delta[v]$ when $\tau$ is clear from the context,
and $\delta_C$ when the transfers are those of a
chain~$C$.  The balance is \textit{gross}: it accounts
for these transfers and for nothing else, in particular
not for execution costs.
\end{definition}

\begin{definition}[Transfer chain]\label{def:chain}
A \textit{transfer chain} $C$ is a sequence of transfers
$t_1, t_2, \ldots, t_k$.  It \textit{originates} at
$s(t_1)$ and \textit{terminates} at $d(t_k)$, and its
\textit{middleman set} is
$\{d(t_1), d(t_2), \ldots, d(t_{k-1})\}$: the addresses
that relay tokens without being the origin or the final
destination.  The chain is \textit{connected} when
consecutive transfers meet, that is $d(t_i) = s(t_{i+1})$
for all $1 \leq i < k$.

Two chains $C_1, C_2$ with $d(C_1) = s(C_2) = j$ join at
$j$ under either \textit{token continuity}, when the last
token of $C_1$ is the first token of $C_2$, or
\textit{balance continuity}, when the junction is left
short in no token:
$\mathit{BalCont}(j) \equiv \forall \tau,\,
\delta_{C_1}[j,\tau] + \delta_{C_2}[j,\tau] \geq 0$.
Balance continuity is the weaker of the two.
\end{definition}

\noindent
We write $s(C)$, $d(C)$ and $M(C)$ for a chain's origin,
destination and middleman set, and lift the token
notation to chains: $\tau_{\mathrm{in}}(C)$ and
$\tau_{\mathrm{out}}(C)$ are the tokens of $t_1$ and
$t_k$, the assets the chain takes in and gives back,
while $\tau_{\mathrm{mid}}(C)$ is the token of
$t_{k-1}$, what the last intermediary passed on.  For a
two-transfer chain $\tau_{\mathrm{mid}}$ and
$\tau_{\mathrm{in}}$ coincide; for longer ones they do
not, and the rules below need both.
In the implementation a chain is a binary tree whose
in-order traversal is the sequence above, preserving the
order of assembly.  Closing a chain into a cycle
requires matching tokens at both ends, but a chain may
start with a native asset and end with its wrapped
form.

\begin{definition}[Token equivalence]
\label{def:tokeq}
Two token types $\tau_1$, $\tau_2$ are
\textit{token-equivalent}, written
$\tau_1 =_\tau \tau_2$, when they are either
identical or denote the same underlying asset
through different representations (e.g., ETH and
WETH).  All other token pairs require strict
equality.  The relation is reflexive and symmetric.
Transitivity is not required, as $=_\tau$ is
checked only at cycle boundaries.
\end{definition}

The last two definitions say what the algorithm looks
for.  Both are stated over the transfer graph alone, so
both can be checked against the transaction's own
transfers.

\begin{definition}[Cycle in $G$]\label{def:cycle}
Let $v$ be an address.  A \textit{cycle at $v$} is a
sequence of connected chains $C_1, \ldots, C_m$ with
$m \geq 1$, whose concatenation $C = C_1 \cdots C_m$
satisfies
\begin{compactenum}
  \item \textit{(maximal)} no two consecutive pieces can
    be joined into one connected chain, so
    $C_1, \ldots, C_m$ is the decomposition of $C$ into
    its longest connected runs;
  \item \textit{(closed)} $C$ originates and terminates
    at $v$;
  \item \textit{(drawn from the transaction)} the
    transfers of $C$, counted with multiplicity, are a
    sub-multiset of $E$.
\end{compactenum}
When $m = 1$ the cycle is a single connected chain, the
classical picture; a larger $m$ means it closes through
parallel legs, with $m - 1$ junctions where consecutive
pieces fail to meet.  The transaction of Appendix~\ref{sec:casestudy:tn}
illustrates the general case: its closed walks are
cycles in this sense, and several close only through
junctions where the token received differs from the
token passed on.
\end{definition}

\begin{definition}[Arbitrage in $G$]\label{def:arb}
A cycle $C$ at $v$ is an \textit{arbitrage in token
$\tau$} when
\begin{compactenum}
  \item \textit{(entry token)} $\tau$ is the token of the
    first transfer of $C$;
  \item \textit{(round trip in one asset)} the first and
    the last transfer of $C$ carry token-equivalent
    tokens;
  \item \textit{(profit)} the balance at the origin is
    strictly positive in that token,
    $\delta_C[v, \tau] > 0$.
\end{compactenum}
No cycle of the transaction in
Appendix~\ref{sec:casestudy:tn} is an arbitrage.  Three
carry token-equivalent tokens at both ends and so
satisfy the first two conditions, but their pieces do
not join: no walk over compatible tokens runs round, and
the balance at the origin is zero.  The same asset
entering and leaving an address is an accounting fact,
not a trade.
\end{definition}

\noindent
Definition~\ref{def:arb} is gross by construction, hence
checkable from the trace alone.  Execution costs are not
part of it: whether a cycle also gained after costs is a
separate query on the algorithm's result.  Requiring the
round trip to close in one asset is the notion of cyclic
arbitrage prior measurement work
uses~\cite{mclaughlin2023,wang2022cyclic}: a walk
returning in a different token leaves the actor holding
something else, and deciding whether that is a gain
needs the price oracle this paper deliberately avoids.

\subsection{Pipeline Overview}\label{sec:arch}

The pipeline has two layers.  The \textit{decode layer} ingests raw traces
(\texttt{debug\_\allowbreak traceTransaction} or a
pre-indexed database).  Transfer extraction uses only
the standard ERC token and \textsc{weth} signatures;
Sourcify ABIs and the 4byte directory serve call-name
and OSINT enrichment, which the fixpoint and the
validation never consume.  It is the only stage
touching external databases.  The
\textit{analysis layer} is pure and in-memory: it builds
the tree of Definition~\ref{def:cft}, rewrites it,
computes balances, and classifies.

\section{Normalization and Detection}\label{sec:algo}

Definition~\ref{def:arb} says what we are looking for.
Finding it by searching the transfer graph means
enumerating cycles in a structure that has already lost
the call hierarchy, which is what lets a chain of
transfers that merely passes through an address look
like one that returns to it.  We work on the tree
instead, rewriting it until the cycles appear as
subterms.

Every object a rule touches is a cash flow tree
(Definition~\ref{def:cft}): a leaf carrying one
transfer, a chain node carrying a labelled sequence, or
a call node carrying an address and its ordered
children.  A rule rewrites the children of one node and
yields a tree of the same kind.  The label records what
the reduction has established about a chain:
\textsc{chaining} for a sequence assembled from adjacent
transfers, \textsc{merging} for one assembled from
parallel branches, \textsc{burn} and \textsc{mint} for
sequences ending at or starting from the null address,
\textsc{cycle} for a closed chain that returns in a
different asset, and \textsc{arb} for one that returns
in the asset it left in.

The rules derive from the chain's own conventions
(swap pairing, null-address mint and burn, wrapping)
and from the composition of walks; R5 is the sole
deployment hint.  The reduction has two rewriting
stages, grouped in Table~\ref{tab:labels} by \emph{locus} rather than by
operation: both stages chain and both merge.
\textit{Leaf manipulation} ($R_1$--$R_{10}$) fires on
siblings within one call frame; when no rule applies at
a level, its chains and leaves are \textit{lifted} into
the enclosing frame, becoming siblings of the call that
produced them, and reduction resumes there.
\textit{Node manipulation} ($R_{11}$--$R_{13}$) fires
once lifting has put operands from different frames side
by side.  Each node rule is the cross-frame counterpart
of a leaf rule with a stricter guard (R11 of R4, R12
of R6, R13 of R7), stricter precisely because lifting
has destroyed the frame evidence the leaf rule could
rely on.  Neither stage names
a query, and neither adds nor drops a transfer.  The
\textit{annotation} rules ($R_{14}$, $R_{15}$) fire
within both stages, and are where the word
\textit{arbitrage} first appears: a chain whose
endpoints coincide becomes \textsc{arb} when it also
leaves and returns in token-equivalent assets, and
\textsc{cycle} otherwise.  Annotation is structural and does not consult amounts,
so \textit{validation} ($R_{16}$) runs once after the
fixpoint and revokes the \textsc{arb} label wherever the
balance at the origin is not positive.

We describe the algorithm through a real Ethereum
transaction~$\mathit{tx}_R$\footnote{%
\texttt{0x\seqsplit{275f9642556a58802eafdd8289aa19a7275e9d40962056095bb9b5c51ac3d246}}} that splits a payment across
parallel paths and recombines them, and that EigenPhi
does not flag as an arbitrage.
Figure~\ref{fig:txR} shows its transfer graph: the
SC splits $0.0575$~\textsc{WETH} into three
parallel paths through six pools, each converting
an intermediate token into \textsc{BEAN}, which
converges on a single \textsc{BEAN/WETH} pool that
returns $0.0598$~\textsc{WETH} (profit:
$+0.0023$~\textsc{WETH}).

\begin{figure*}[t]
\centering
\resizebox{0.79\textwidth}{!}{%
\begin{tikzpicture}[font=\scriptsize,
  acct/.style={draw=gray!70, fill=acctgray!30,
    rounded rectangle, minimum height=6mm,
    font=\scriptsize, inner sep=3pt,
    align=center},
  sc/.style={draw=scamber!80, fill=scamber!25,
    rounded rectangle, minimum height=6mm,
    font=\scriptsize, inner sep=3pt,
    align=center,
    double, double distance=0.8pt},
  pool/.style={draw=poolblue!80, fill=poolblue!15,
    rounded corners, minimum height=6mm,
    font=\scriptsize, inner sep=3pt},
  lbl/.style={font=\scriptsize, fill=white,
    inner sep=0.5pt},
  arr/.style={-{Stealth[length=2.5mm]}, semithick},
]
% --- main actors (left) ---
\node[acct] (from) at (0,4.5) {\textsf{from}\,(EOA)};
\node[acct] (builder) at (3.5,4.5) {blockBuilder};
\node[sc]   (to) at (0,2.2) {\textsf{to}\,(SC)};

% --- hop-1 pools (center-left) ---
\node[pool] (rweth) at (5,3.3)
  {\textsc{chad/weth}};
\node[pool] (pweth) at (5,2.2)
  {\textsc{duck/weth}};
\node[pool] (wmoon) at (5,1.1)
  {\textsc{weth/moon}};

% --- hop-2 pools (center-right) ---
\node[pool] (cbean) at (9.5,3.3)
  {\textsc{chad/bean}};
\node[pool] (dbean) at (9.5,2.2)
  {\textsc{duck/bean}};
\node[pool] (beanmoon) at (9.5,1.1)
  {\textsc{bean/moon}};

% --- convergence pool (right) ---
\node[pool] (beanweth) at (13.5,2.2)
  {\textsc{bean/weth}};

% === edges ===
% cost (dashed)
\draw[arr, densely dashed]
  (from) -- node[lbl,above] {ETH} (builder);
\draw[arr, densely dashed]
  (from) -- node[lbl, left, pos=0.4]
  {ETH} (to);

% to → three WETH paths (blue)
\draw[arr, okBlue]
  (to) -- node[lbl, above, sloped, pos=0.45]
  {0.0125 \textsc{weth}} (rweth);
\draw[arr, okBlue]
  (to) -- node[lbl, above, pos=0.45]
  {0.0333 \textsc{weth}} (pweth);
\draw[arr, okBlue]
  (to) -- node[lbl, below, sloped, pos=0.45]
  {0.0117 \textsc{weth}} (wmoon);

% hop 1 → hop 2 (red)
\draw[arr, okVermillion]
  (rweth) -- node[lbl,above]
  {118.5 \textsc{chad}} (cbean);
\draw[arr, okVermillion]
  (pweth) -- node[lbl,above]
  {805M \textsc{duck}} (dbean);
\draw[arr, okVermillion]
  (wmoon) -- node[lbl,above]
  {171K \textsc{moon}} (beanmoon);

% hop 2 → BEAN/WETH (green)
\draw[arr, okGreen]
  (cbean) -- node[lbl, above, sloped, pos=0.5]
  {382M \textsc{bean}} (beanweth);
\draw[arr, okGreen]
  (dbean) -- node[lbl, above, pos=0.5]
  {970M \textsc{bean}} (beanweth);
\draw[arr, okGreen]
  (beanmoon) -- node[lbl, below, sloped, pos=0.5]
  {353M \textsc{bean}} (beanweth);

% BEAN/WETH → to (WETH return, blue thick, routed below)
\draw[arr, okBlue]
  (beanweth.south) -- (13.5,0.45)
  -- node[lbl, below] {0.0598 \textsc{weth}}
  (0,0.45) -- (to.south);

\end{tikzpicture}%
}
\caption{Transfer graph of~$\mathit{tx}_R$.
  Three parallel paths through six pools;
  dashed edges are cost transfers.}
\label{fig:txR}
\end{figure*}
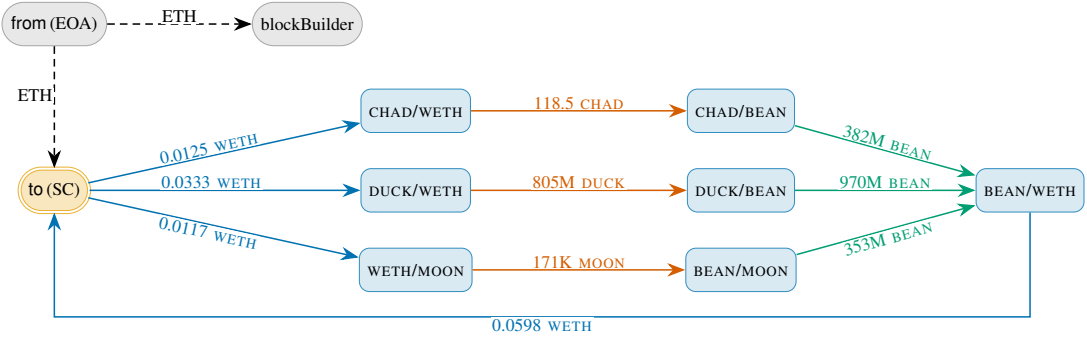

\subsection{From Traces to ASTs}
\label{sec:algo:ast}

\begin{figure*}[t]
\centering
\tikzset{
  nd/.style={draw, circle, fill=gray!15,
    inner sep=1.5pt, font=\footnotesize,
    minimum size=5mm},
  lf/.style={draw, rectangle, inner sep=2pt,
    font=\scriptsize, minimum height=5.5mm,
    align=center},
  ch/.style={draw, double, double distance=0.8pt,
    rectangle, rounded corners=2pt,
    inner sep=2.5pt, font=\footnotesize,
    minimum height=6mm, align=center},
  mg/.style={draw, double, double distance=0.8pt,
    rectangle, rounded corners=2pt,
    inner sep=2.5pt, font=\footnotesize,
    minimum height=6mm, align=center,
    fill=okGreen!10},
  cy/.style={draw, line width=0.9pt,
    double, double distance=0.8pt,
    rectangle, rounded corners=2.5pt,
    inner sep=3pt, font=\footnotesize,
    minimum height=7mm, align=center,
    fill=okBlue!8},
  e/.style={draw=gray!50, thin},
}
% ════════ (a) Trimmed tree: the callback cascade ════════
\begin{subfigure}[t]{0.52\textwidth}
\centering
\begin{tikzpicture}[font=\scriptsize, scale=1.055,
  every node/.style={anchor=west, inner sep=1.6pt,
    font=\scriptsize},
  fr/.style={draw, rounded corners=2pt, fill=gray!12,
    minimum height=3.8mm},
  lf2/.style={draw, minimum height=3.8mm},
  ed/.style={draw=gray!60, thin}]
\node[fr]  (n0) at (0.00, 0.00) {EOA$\Rightarrow$SC};
\node[fr]  (n1) at (0.45,-0.44) {SC$\Rightarrow$BN/W};
\node[fr]  (n2) at (0.90,-0.88) {BN/W$\Rightarrow$SC};
\node[fr]  (n3) at (1.35,-1.32) {SC$\Rightarrow$C/BN};
\node[fr]  (n4) at (1.80,-1.76) {C/BN$\Rightarrow$SC};
\node[fr]  (n5) at (2.25,-2.20) {SC$\Rightarrow$C/W};
\node[lf2, fill=okOrange!25] (n6) at (2.70,-2.64)
  {C/W$\to$C/BN \textsc{chad}};
\node[lf2, fill=okBlue!15]   (n7) at (2.25,-3.08)
  {SC$\to$C/W \textsc{weth}};
\node[draw, dashed, gray!70, rounded corners=2pt,
      minimum height=3.8mm] (n8) at (2.25,-3.52)
  {\textcolor{gray!85}{$\cdots$ six more levels}};
\node[lf2, fill=okGreen!15]  (n9) at (1.80,-3.96)
  {C/BN$\to$BN/W \textsc{bean}};
\node[lf2, fill=okBlue!15]  (n10) at (0.90,-4.40)
  {BN/W$\to$SC \textsc{weth}};
\foreach \p/\c in {n0/n1,n1/n2,n2/n3,n3/n4,n4/n5,n5/n6,
                   n4/n7,n4/n8,n3/n9,n1/n10}
  \draw[ed] ([xshift=1.6mm]\p.south west) |- (\c.west);
\draw[gray!70] (4.95,-2.50) -- (5.09,-2.50) --
               (5.09,-3.22) -- (4.95,-3.22);
\node[anchor=west] at (5.15,-2.86)
  {\textcircled{\scriptsize 1}\,lift, R1};
\draw[gray!70] (4.95,-3.22) -- (5.09,-3.22) --
               (5.09,-4.10) -- (4.95,-4.10);
\node[anchor=west] at (5.15,-3.66)
  {\textcircled{\scriptsize 2}\,lift, R6};
\end{tikzpicture}
\caption{The trimmed tree, drawn down the \textsc{chad}
  path only.  $\Rightarrow$~is a call frame, $\to$~a
  transfer; the dashed node elides the rest.}
\label{fig:ast:a}
\end{subfigure}%
\hfill
% ════════ (b) The three states that follow ════════
\begin{subfigure}[t]{0.44\textwidth}
\centering
\begin{tikzpicture}[font=\scriptsize, scale=1.135,
  every node/.style={anchor=west, inner sep=1.6pt,
    font=\scriptsize},
  fr/.style={draw, rounded corners=2pt, fill=gray!12,
    minimum height=3.8mm},
  lf2/.style={draw, minimum height=3.8mm},
  cn/.style={draw, double, double distance=0.7pt,
    rounded corners=2pt, minimum height=3.8mm},
  ed/.style={draw=gray!60, thin}]
\node[fr]  (m0) at (0.00, 0.00) {SC$\Rightarrow$BN/W};
\node[lf2, fill=okBlue!15] (m1) at (0.45,-0.46)
  {BN/W$\to$SC \textsc{weth}};
\node[cn,  fill=okGreen!10] (m2) at (0.45,-0.92)
  {SC$\to$BN/W \emph{3 paths merged}};
\draw[ed] ([xshift=1.6mm]m0.south west) |- (m1.west);
\draw[ed] ([xshift=1.6mm]m0.south west) |- (m2.west);
\draw[-latex, gray!70] (2.1,-1.24) -- (2.1,-1.72);
\node[anchor=west] at (2.20,-1.48) {R6 closes the cycle};
\node[cn, fill=okBlue!8] (m3) at (0.45,-2.04)
  {SC$\to$SC \emph{chaining}};
\draw[-latex, gray!70] (2.1,-2.36) -- (2.1,-2.84);
\node[anchor=west] at (2.20,-2.60) {R14 annotates};
\node[cn, line width=0.8pt, fill=okBlue!8] (m4) at (0.45,-3.16)
  {SC$\to$SC \emph{arbitrage} $\Delta={+}0.0023$~\textsc{weth}};
\end{tikzpicture}
\caption{After the structural fixpoint's merge: R6
  closes the aggregate, R14 annotates it, and
  validation confirms a positive balance at the
  origin.}
\label{fig:ast:b}
\end{subfigure}
\caption{Reduction of~$\mathit{tx}_R$, a callback
  cascade in which lifting creates every chaining
  opportunity.  Double borders mark chains.}
\label{fig:ast}
\end{figure*}
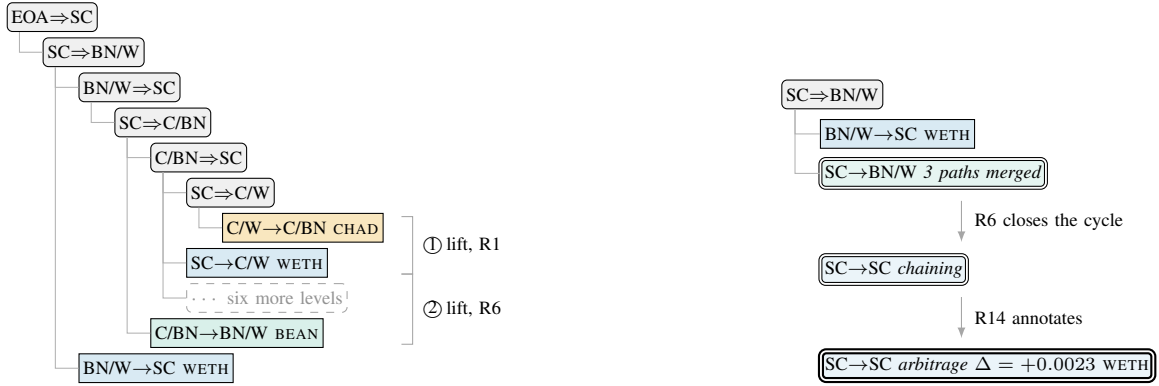

Before the stages begin, the AST of
Definition~\ref{def:cft} is built from the decoded
trace: one $\mathsf{Tree}$ node per call frame, one
$\mathsf{Leaf}$ per value-movement event
(\texttt{Transfer}, \texttt{Transfer\-Single},
\texttt{Transfer\-Batch}, \texttt{Mint}, \texttt{Burn},
\texttt{Deposit}, \texttt{Withdrawal}, native~ETH),
with $\sigma$ taken from the enclosing frame as in
\S\ref{sec:bg:exec}.  Non-value events such as
\texttt{Approval} are discarded.  Each decoded transfer
becomes exactly one leaf, so the tree carries all of
them and no others, and the nesting is the call
hierarchy.  That the trace admits this reading at all
rests on a design invariant of blockchain virtual
machines:

\begin{property}[Deterministic sequential execution]\label{prop:dse}
Let $\mathit{tx}$ be a transaction executed on a
blockchain in state~$S$.  The execution trace
$\mathcal{T}(\mathit{tx}, S)$ satisfies:
\begin{enumerate}
  \item \textbf{Determinism.}
    $\mathcal{T}$ is a function: the trace is
    uniquely determined by $(\mathit{tx}, S)$.
  \item \textbf{Sequential order.}
    Within each call frame, instructions execute
    in program order.  No concurrent or
    out-of-order execution occurs within a frame.
  \item \textbf{Tree structure.}
    Nested calls (\texttt{CALL},
    \texttt{DELEGATE\-CALL},
    \texttt{STATIC\-CALL}) create child frames.
    The resulting call hierarchy is a rooted,
    ordered tree.
\end{enumerate}
\end{property}

\noindent
None of the three is an empirical observation:
consensus requires deterministic state transitions,
sequential order within a frame follows from the
stack-machine architecture, and the tree structure from
the call/return discipline.  A machine violating any of
them could still reach consensus, as parallel-execution
VMs do by scheduling conflict-free frames concurrently,
but not under the ordered-tree view we exploit.
EVM-compatible chains satisfy
Property~\ref{prop:dse} by construction, and it is the
only assumption our results make about the execution
environment; we validate on two further EVM chains
(Appendix~\ref{sec:arbitrum}).  Non-EVM machines are
future work.

\paragraph{Trimming}
Trimming then removes structure that carries no
transfer.  A $\mathsf{Tree}$ node whose children are all
non-transfer events is pruned; a node with exactly one
child is collapsed into that child; and any subtree
whose recursive pruning yields no $\mathsf{Leaf}$ is
removed.  A second cleanup drops degenerate leaves,
transfers of zero amount with $\sigma = s(\ell)$, which
move no value and would otherwise block chaining at the
shared address.
Construction and trimming together produce the tree
$T_0$ on which the stages operate.  Every guarantee in
\S\ref{sec:algo:formal} is stated from $T_0$: the
theorems say what the rewriting does to the transfers it
is given, not whether those transfers faithfully reflect
the value that moved: that fidelity is the decode
layer's, and belongs to the trusted base
(\S\ref{sec:limits}).
For~$\mathit{tx}_R$ the raw trace holds fifty-one call
frames, which trimming collapses to twelve.  The shape
survives: each pool calls back into the SC before
settling, so the frames nest twelve deep and each of the
ten route transfers sits alone in its own frame, none
adjacent to another (Figure~\ref{fig:ast}(a)).  The two
dashed transfers of Figure~\ref{fig:txR} are native-ETH
costs and never enter $T_0$; the profit query accounts
for them after the fixpoint
(\S\ref{sec:algo:leftover}).  Panel~(b) follows the
reduction from there.

\subsection{Stage 1: Leaf Manipulation}
\label{sec:algo:leaves}

\begin{algorithm}[t]
\small
\begin{tabbing}
\hspace{4mm}\=\hspace{4mm}\=\hspace{4mm}\=%
  \hspace{4mm}\=\kill
\textbf{function} \textsc{Manipulate-Leaves}($T$):\\
\> \textbf{for each} level $\ell$ from deepest
  to root \textbf{do}\\
\>\> \textbf{repeat}\\
\>\>\> \textbf{for each} sibling pair
  $(\ell_i, \ell_{i+1})$ at level $\ell$\\
\>\>\>\> \textbf{if}
  $d(\ell_i) = s(\ell_{i+1})$ and
  $\tau(\ell_i) \neq \tau(\ell_{i+1})$\\
\>\>\>\> \textbf{then} replace
  $(\ell_i, \ell_{i+1})$ with
  \textsc{Chain}$(\ell_i, \ell_{i+1})$\\
\>\>\> \textbf{for each}
  $\mathsf{Tree}(c, [t_1 \ldots t_m])$
  at level $\ell$\\
\>\>\>\> \textbf{if} all $t_j$ are chains
  or leaves\\
\>\>\>\> \textbf{then} lift $t_1 \ldots t_m$
  to parent\\
\>\> \textbf{until} no change at level $\ell$\\
\> \textbf{return} $T$
\end{tabbing}
\caption{Leaf manipulation.
  \textsc{Chain} constructs a transfer chain
  from two adjacent transfers.  Lifting dissolves fully-reduced
  $\mathsf{Tree}$ nodes.}
\label{alg:leaves}
\end{algorithm}

\begin{table}[t]
\caption{Rewriting rules.  $\ell$:~leaf;
  $C_{\textit{label}}$:~labelled chain.
  $^\dagger$Optional deployment hint.}
\label{tab:labels}
\centering
{\footnotesize
\begin{tabular}{@{}@{\hspace{2pt}}c@{\hspace{4pt}}>{\raggedright\arraybackslash}p{5.6cm}@{\hspace{3pt}}l@{}}
\toprule
\textcolor{darkred}{\#} & Premise & \textcolor{darkred}{Conclusion} \\
\midrule
\multicolumn{3}{@{}l}{\textbf{Leaf manipulation}
  (within a call-frame node)} \\[2pt]
\textcolor{darkred}{R1}
  & $d(\ell_1){=}s(\ell_2)$,
    $\tau(\ell_1){\neq}\tau(\ell_2)$,
    $\mathit{burn}(\ell_1){\Leftrightarrow}\mathit{mint}(\ell_2)$
  & \textcolor{darkred}{$(\ell_1, \ell_2) \to C_{\textit{chain}}$} \\
\textcolor{darkred}{R2}
  & $\ell_b$ burn, $\ell$ not mint,
    $d(\ell_b){=}s(\ell)$
  & \textcolor{darkred}{$(\ell_b, \ell) \to C_{\textit{burn}}$} \\
\textcolor{darkred}{R3}
  & $\ell_m$ mint, $\ell$ not burn,
    $d(\ell){=}s(\ell_m)$
  & \textcolor{darkred}{$(\ell, \ell_m) \to C_{\textit{mint}}$} \\
\textcolor{darkred}{R4}
  & $d(\ell_1){=}s(\ell_2)$, $d(\ell_2){=}s(\ell_1)$,
    $\tau(\ell_1){=}\tau(\ell_2)$,
    $d(\ell_1) \notin \mathit{Routers}$,
    $\exists\, i,\, \sigma(\ell_i){\neq}d(\ell_1)$
  & \textcolor{darkred}{$(\ell_1, \ell_2) \to C_{\textit{chain}}$} \\
\textcolor{darkred}{R5}$^\dagger$
  & $d(\ell_1){=}s(\ell_2)$, $\tau(\ell_1){=}\tau(\ell_2)$,
    $d(\ell_1) \in \mathit{Routers}$,
    $\mathit{burn}(\ell_1){\Leftrightarrow}\mathit{mint}(\ell_2)$
  & \textcolor{darkred}{$(\ell_1, \ell_2) \to C_{\textit{chain}}$} \\
\textcolor{darkred}{R6}
  & $d(C){=}s(\ell) \vee d(\ell){=}s(C)$
  & \textcolor{darkred}{$(\ell, C) \to C'_{\textit{chain}}$} \\
\textcolor{darkred}{R7}
  & $s(C_1){=}s(C_2)$, $d(C_1){=}d(C_2)$, $s{\neq}d$,
    $\tau_{\mathrm{out}}(C_1){=}\tau_{\mathrm{out}}(C_2)$,
    $\tau_{\mathrm{mid}}(C_1){\neq}\tau_{\mathrm{mid}}(C_2)$
  & \textcolor{darkred}{$(C_1, C_2) \to C_{\textit{merge}}$} \\
\textcolor{darkred}{R8}
  & $s(C_1){=}s(C_2){=}d(C_1){=}d(C_2)$,
    $[\tau_{\mathrm{in}}(C_1){=}\tau_{\mathrm{in}}(C_2)
    \wedge \tau_{\mathrm{mid}}(C_1){=}\tau_{\mathrm{mid}}(C_2)
    \wedge \tau_{\mathrm{out}}(C_1){=}\tau_{\mathrm{out}}(C_2)]
    \vee \mathit{BalCont}(s)$
  & \textcolor{darkred}{$(C_1, C_2) \to C_{\textit{merge}}$} \\
\textcolor{darkred}{R9}
  & $s(C_1){=}d(C_1){=}s(C_2){=}d(C_2)$,
    $\tau_{\mathrm{in}}(C_1){=}\tau_{\mathrm{in}}(C_2)$,
    $\tau_{\mathrm{out}}(C_1){=}\tau_{\mathrm{out}}(C_2)$
  & \textcolor{darkred}{$(C_1, C_2) \to C_{\textit{merge}}$} \\
\textcolor{darkred}{R10}
  & $d(C_1){=}s(C_2)$,
    $[\tau_{\mathrm{out}}(C_1){=}\tau_{\mathrm{in}}(C_2)]
    \vee \mathit{BalCont}(d(C_1))$
  & \textcolor{darkred}{$(C_1, C_2) \to C_{\textit{chain}}$} \\
\midrule
\multicolumn{3}{@{}l}{\textbf{Node manipulation}
  (across call frames)} \\[2pt]
\textcolor{darkred}{R11}
  & $d(\ell_1){=}s(\ell_2)$,
    $\tau(\ell_1){=}\tau(\ell_2)$,
    $d(\ell_2){\neq}s(\ell_1)$,
    $d(\ell_1) \notin \mathit{Routers}$,
    $\sigma(\ell_1) \neq d(\ell_1)$
  & \textcolor{darkred}{$(\ell_1, \ell_2) \to C_{\textit{chain}}$} \\
\textcolor{darkred}{R12}
  & $[d(C){=}s(\ell) \wedge \tau_{\mathrm{out}}(C){=}\tau(\ell)]
    \vee [d(\ell){=}s(C) \wedge \tau(\ell){=}\tau_{\mathrm{in}}(C)]$
  & \textcolor{darkred}{$(\ell, C) \to C'_{\textit{chain}}$} \\
\textcolor{darkred}{R13}
  & R7's premises, and
    $\tau_{\mathrm{in}}(C_1){=}\tau_{\mathrm{in}}(C_2)$
  & \textcolor{darkred}{$(C_1, C_2) \to C_{\textit{merge}}$} \\
\midrule
\multicolumn{3}{@{}l}{\textbf{Annotation}
  (both leaf and node stages)} \\[2pt]
\textcolor{darkred}{R14}
  & $s(C){=}d(C)$, $\tau_{\mathrm{in}}(C){=_\tau}\tau_{\mathrm{out}}(C)$,
    $[\sigma(C){\notin}C \vee s(C){=}\sigma(C)]$,
    $\neg\mathit{WrapUnwrap}(C)$
  & \textcolor{darkred}{$C_{\textit{chain}} \to C_{\textit{arb}}$} \\
\textcolor{darkred}{R15}
  & $s(C){=}d(C)$,
    $\sigma(C) \in C$
  & \textcolor{darkred}{$C_{\textit{chain}} \to C_{\textit{cycle}}$} \\
\midrule
\multicolumn{3}{@{}l}{\textbf{Validation}
  (post-annotation)} \\[2pt]
\textcolor{darkred}{R16}
  & $\delta_C[s(C), \tau_{\mathrm{in}}(C)] \leq 0$
  & \textcolor{darkred}{$C_{\textit{arb}} \to C_{\textit{cycle}}$} \\
\bottomrule
\end{tabular}}
\end{table}

Leaf manipulation, the first rewriting stage
(R1--R10), operates on sibling children within
each $\mathsf{Tree}$ node, bottom-up from the deepest
subtrees to the root.  Algorithm~\ref{alg:leaves} gives
the pseudocode, Table~\ref{tab:labels} the rules.

\paragraph{Chaining}
Two sibling leaves are \textit{chainable} when
$d(\ell_1) = s(\ell_2)$ and
$\tau(\ell_1) \neq \tau(\ell_2)$ (R1).  The address
match locates a shared intermediary, typically a pool;
the token mismatch is what a swap looks like from
outside, one asset in and another out.  The pair
becomes a chain with origin~$s(\ell_1)$,
destination~$d(\ell_2)$ and middleman~$\{d(\ell_1)\}$.
The rule extends transitively, to a chain and an
adjacent leaf (R6) and to two chains that meet (R10);
R10 also fires under balance continuity at the
junction (Definition~\ref{def:chain}), which admits
balance-preserving routing such as a
WETH$\leftrightarrow$USDC relay.

Chains that run in parallel rather than in sequence are
merged instead.  R7 takes two chains with the same
origin and destination that give back the same asset
but pass through different intermediaries: the premise
$\tau_{\mathrm{mid}}(C_1) \neq \tau_{\mathrm{mid}}(C_2)$
says exactly that, and without it the rule would merge
a route with itself.  R8 and R9 are its closed forms,
for chains that already return to their origin, R8
requiring the intermediate token to agree as well so
that merging is summation.  A merged node is labeled
\textit{merging} rather than \textit{chaining}: a
chained node extends one path, a merged node holds
several.  Either way the delta maps add,
$\delta_{C_m}[v, \tau] = \delta_{C_1}[v, \tau] +
\delta_{C_2}[v, \tau]$ (Definition~\ref{def:balance}).

Chaining a round trip needs one further decision.  When
two sibling leaves return to where they started through
an address~$B$ (R4), $B$ is either swapping, in which
case the pair is a chain, or merely relaying, in which
case chaining it would short-circuit the path.  We read
this off the \textit{sender} field: when $\sigma \neq B$
on either leaf, an external contract is driving the
interaction \textit{through} $B$, the structural
signature of a swap.  The check needs no external
metadata, only the call hierarchy already in the trace.
R4 yields a plain chain; whether the round trip
deserves a name is decided later, by annotation.
Deployments may optionally declare router addresses
(R5, the only protocol-specific rule) so that
same-token transfers through an aggregator chain
anyway; disabling the hint costs recall on
universal-router transactions and adds no false
positives, since removing a rule only removes
reductions.  Ours holds one address, the Uniswap~V4
universal router.

\paragraph{Mint- and burn-aware chaining}
A transfer out of the null address
(\texttt{0x0000\ldots0}) is a \textit{mint} and one
into it a \textit{burn}; the decode layer normalizes
both to name the emitting contract instead.  Since the
null address is not a participant, neither can chain on
its own, so each pairs with an adjacent ordinary
transfer that completes the flow.  A burn followed by a
transfer out of the same contract is a \textit{token
burn} (R2), value leaving as the token is destroyed,
which also covers Withdrawal\,+\,NativeTransfer
(unwrap); a transfer followed by a mint is the mirror
(R3), covering Deposit\,+\,NativeTransfer (wrap).  The
two exclude each other, so neither fires on a burn
feeding a mint, a token migration rather than either,
which R1 chains plainly.  Token burns and mints
then take part in cycle connection on the same footing
as \textit{cycle}-labeled chains, so the fixpoint
reaches arbitrages routed through wrapped-token
contracts and vault deposits.  The normalization
belongs to the decode layer; the algorithm treats the
label as opaque.

\paragraph{Lifting}
Once an internal $\mathsf{Tree}$ node holds only chains
and leaves, its children are \textit{lifted} into the
parent, becoming siblings of the frame that produced
them.  This exposes transfers hidden inside nested
calls for chaining one level up.  Only internal frames
are lifted; the root is never dissolved.

\paragraph{Bottom-up iteration}
Chaining and lifting run as a fixpoint at each level,
deepest first.  In $\mathit{tx}_R$ lifting is not
housekeeping but the precondition for everything else:
in Figure~\ref{fig:ast}(a) no two of the ten transfers
are siblings,
so no rule can fire at all until frames are dissolved.
The \textsc{WETH} leaf $\mathit{SC} \to \mathit{C/W}$
sits beside the frame holding
$\mathit{C/W} \to \mathit{C/BN}$ in \textsc{CHAD}; one
lift makes the two leaves siblings and R1 chains them
into $\mathit{SC} \to \mathit{C/BN}$
(mark~\textcircled{\scriptsize 1}).  A second lift
brings the \textsc{BEAN} leaf
$\mathit{C/BN} \to \mathit{BN/W}$ alongside, and R6
extends the chain to $\mathit{SC} \to \mathit{BN/W}$
(mark~\textcircled{\scriptsize 2}).
Twelve levels of this leave
three chains agreeing at both boundaries but differing
in the middle ($\tau_{\mathrm{mid}} =$ \textsc{CHAD},
\textsc{DUCK}, \textsc{MOON}) beside the return leaf
$\mathit{BN/W} \to \mathit{SC}$, ready for the merge.

\subsection{Stage 2: Node Manipulation}
\label{sec:algo:nodes}

\begin{algorithm}[t]
\small
\begin{tabbing}
\hspace{4mm}\=\hspace{4mm}\=\hspace{4mm}\=%
  \hspace{4mm}\=\kill
\textbf{function} \textsc{Annotate-and-Reduce}%
  ($T$, $\mathit{from}$, $\mathit{to}$):\\[2pt]
\> \textrm{// Label closed chains
  (Def.~\ref{def:cycle})}\\
\> $T' \leftarrow$ \textsc{Annotate-Cycles}%
  ($T$, $\mathit{from}$)\\[2pt]
\> \textrm{// Merge complementary open chains}\\
\> $T'' \leftarrow$ \textsc{Connect-Cycles}%
  ($T'$, $\mathit{from}$, $\mathit{to}$)\\[2pt]
\> \textbf{if} $T'' = T$ \textbf{then return} $T$\\
\> \textbf{else return}
  \textsc{Annotate-and-Reduce}%
  ($T''$, $\mathit{from}$, $\mathit{to}$)
\end{tabbing}
\caption{Annotation and connection loop; runs after
  the structural fixpoint.}
\label{alg:fixpoint}
\end{algorithm}

After leaf manipulation, the root
$\mathsf{Tree}$ node's children are a mixture of
transfer chains (Definition~\ref{def:chain}) and
residual $\mathsf{Leaf}$ nodes.
We write $a \in C$ for $a \in M(C)$.
The node-level pass first chains residual elements
using a \emph{same-token} rule~(R11), the dual of
the leaf-level different-token rule: it captures
the same asset flowing across contract
boundaries.  Likewise, a leaf and a chain are
chained when address-adjacent and
token-compatible~(R12).
Each address absorbed this way joins the chain's
middleman set~$M$; in practice these are routers,
aggregators and proxies, recognized structurally from
the pass-through pattern rather than from any list.
Parallel chains are merged as at the leaf level, by
R13: R7's premises with the entry tokens pinned as
well, which the node level can require because both
origins are frame actors and their entry tokens are
therefore determined.
The annotation rules ($R_{14}$, $R_{15}$) then label
the chains that close.
Algorithm~\ref{alg:fixpoint} gives the pseudocode
and Table~\ref{tab:labels} the rewriting rules.

\paragraph{Cycle annotation}
Cycle annotation labels each closed chain
($s(C) = d(C)$) with one of two labels (R14--R15).
Let $\sigma(C)$ be the $\sigma$ of the leaf at
$s(C)$; when ambiguous, $\sigma(C)$ is taken from
the eldest ancestor frame covering the operands,
so that $\sigma$ tracks the call frame that
authored the chain rather than the bottom-up
trace order.
The two rows part on $\sigma(C)$: R14 wants it outside
the middleman set (a bot orchestrating swaps) or equal
to the origin (an EOA-initiated cycle), R15 finds it
inside, which is routing rather than profit extraction.
R14's remaining guard excludes pure native/wrapped
roundtrips: $\mathit{WrapUnwrap}(C)$ holds when $M(C)$
contains a contract of
$\tau_{\mathrm{in}}/\tau_{\mathrm{out}}$ and both
sub-chains are leaves or both satisfy it.
Annotation does not inspect the delta map; the
economic check is $R_{16}$'s
(\S\ref{sec:algo:leftover}).

\paragraph{Cycle connection}
Not all cycles form within a single subtree: a
transaction may leave two labeled chains that close
only when taken together, one carrying the outward leg
and one the return.
The \textit{cycle connection} rule identifies
pairs of already-labeled chains $C_1$, $C_2$
(labeled \textit{cycle}, \textit{token burn},
or \textit{token mint}) at the same tree level
such that
(i)~$s(C_1) = s(C_2) = \mathit{to}$
(same origin, equal to the transaction recipient),
(ii)~$\tau_{\mathrm{out}}(C_1) =
\tau_{\mathrm{in}}(C_2)$ (complementary token
flows), and
(iii)~$s(C_1 \cdot C_2) = d(C_1 \cdot C_2)$
(the concatenation forms a cycle).
An additional filter excludes pairs where either
chain's middleman set is reduced to the sender
alone ($M(C_i) = \{\sigma(C_i)\}$): a routing
pattern (cf.\ R15).
When such a pair is found, the two chains are merged
into a single cycle node.  Connection uses strict
token equality; annotation uses $=_\tau$
(Definition~\ref{def:tokeq}) to label the resulting
cycle.  This separation is deliberate: connection
matches exact token flows, while annotation
recognizes that a native-asset transfer followed by
a wrapped-asset chain (or vice versa) closes the
same economic cycle.  $=_\tau$ is a
representation-level fact about the chain, not about
any specific protocol.

\paragraph{Termination}
Each pass of \textsc{Annotate-and-Reduce} either
(i)~labels at least one chain as \textit{arbitrage}
or \textit{cycle}, or (ii)~merges two distinct chains
into one, strictly reducing the child count of some
$\mathsf{Tree}$ node.  Since labels are monotonic
(a chain is never re-labeled from \textit{arbitrage}
back to \textit{chaining}) and the node count is
finite, the loop terminates in at most
$3n - 2$~passes, where $n$ is the initial
number of $\mathsf{Leaf}$ nodes
(Theorem~\ref{thm:term}).  That bounds
\emph{passes}; what a pass costs is an
implementation matter, priced and measured in
\S\ref{sec:eval:perf}.  Deep, narrow trees converge
fast because each subtree is small; wide, flat ones
approach the pass bound.

In~$\mathit{tx}_R$ the three chains left by the leaf
stage share origin~$\mathit{SC}$ and
destination~$\mathit{BN/W}$ and agree in both boundary
tokens, so R13 merges them into one aggregate, and the
return leaf $\mathit{BN/W} \to \mathit{SC}$ closes that
aggregate into a cycle by R6.  Both fire in the
structural fixpoint, which runs the rules of \emph{both}
loci to exhaustion (a merge creates new adjacencies,
and chaining consumes them) before any label is
assigned.  Algorithm~\ref{alg:fixpoint} is the loop that
follows: here it only names the cycle
(Figure~\ref{fig:ast}(b)).

\subsection{After the Fixpoint: Validation and Profit}
\label{sec:algo:leftover}

\begin{algorithm}[t]
\small
\begin{tabbing}
\hspace{4mm}\=\hspace{4mm}\=\hspace{4mm}\=\kill
\textbf{function} \textsc{Detect-Arbitrage}%
  ($\mathit{trace}$, $\mathit{from}$,
  $\mathit{to}$, $\mathit{meta}$):\\[2pt]
\> \textrm{// Build (\S\ref{sec:algo:ast})}\\
\> $T \leftarrow$ \textsc{Trim}%
  (\textsc{Build-AST}($\mathit{trace}$))\\[2pt]
\> \textrm{// Both loci (\S\ref{sec:algo:leaves},
  \S\ref{sec:algo:nodes}):
  $R_1$--$R_{13}$}\\
\> $T \leftarrow$
  \textsc{Structural-Fixpoint}($T$)\\[2pt]
\> \textrm{// Annotation (\S\ref{sec:algo:nodes}):
  $R_{14}$, $R_{15}$, connection}\\
\> $T \leftarrow$
  \textsc{Annotate-and-Reduce}%
  ($T$, $\mathit{from}$, $\mathit{to}$)\\[2pt]
\> \textrm{// Validation (\S\ref{sec:algo:leftover}):
  $R_{16}$}\\
\> $T \leftarrow$
  \textsc{Validate-Deltas}($T$, $\mathit{to}$)\\[2pt]
\> \textrm{// Profit (\S\ref{sec:algo:leftover}),
  verdict (\S\ref{sec:classify})}\\
\> $(C, L) \leftarrow$
  \textsc{Extract-and-Recover}($T$,
  $\mathit{to}$)\\
\> $\Delta \leftarrow$
  \textsc{Profit}($C$, $L$, $\mathit{to}$,
  $\mathit{meta}$)\\
\> \textbf{return}
  \textsc{Classify}($C$, $L$, $\Delta$)
\end{tabbing}
\caption{Complete detection pipeline.}
\label{alg:detect}
\end{algorithm}

\begin{algorithm}[t]
\small
\begin{tabbing}
\hspace{4mm}\=\hspace{4mm}\=\hspace{4mm}\=%
  \hspace{4mm}\=\kill
\textbf{function} \textsc{Classify}%
  ($C$, $L$, $\Delta$):\\[2pt]
\> $R \leftarrow \emptyset$
  \hfill\textrm{// reason accumulator}\\
\> \textbf{if} $C = \emptyset$
  \textbf{then} $R \leftarrow R \cup
  \{\textsc{no\_cycles}\}$\\
\> \textbf{if} $\Delta_{\mathrm{gross}}$
  is mixed \textbf{then}
  $R \leftarrow R \cup
  \{\textsc{balance\_mixed}\}$\\
\> \textbf{if} net profit $\leq 0$
  \textbf{then}
  $R \leftarrow R \cup
  \{\textsc{neg\_profit}\}$\\
\> \textbf{if} $\Delta_{\mathrm{net}}$
  is mixed \textbf{then}
  $R \leftarrow R \cup
  \{\textsc{final\_mixed}\}$\\
\> \textbf{if} $\Delta_{\mathrm{net}}$
  is negative \textbf{then}
  $R \leftarrow R \cup
  \{\textsc{final\_neg}\}$\\
\> \textbf{if} $L \neq \emptyset$
  \textbf{then}
  $R \leftarrow R \cup
  \{\textsc{leftovers}\}$\\[4pt]
\> \textrm{// Verdict (strict priority)}\\
\> \textbf{if}
  $\textsc{no\_cycles} \in R$
  \textbf{then return}
  (\textit{None}, $R$)\\
\> \textbf{if}
  $\textsc{leftovers} \in R$
  \textbf{then return}
  (\textit{Warning}, $R$)\\
\> \textbf{if}
  $\textsc{final\_neg} \in R$
  \textbf{then return}
  (\textit{Warning}, $R$)\\
\> \textbf{if}
  $\textsc{final\_mixed} \in R$
  \textbf{then return}
  (\textit{Warning}, $R$)\\
\> \textbf{return}
  (\textit{Arbitrage}, $R$)
\end{tabbing}
\caption{Classification logic.}
\label{alg:classify}
\end{algorithm}

Once the fixpoint terminates, the reduced AST
contains annotated cycles, open transfer chains
that did not close,
and residual $\mathsf{Leaf}$ nodes, the transfers the
reduction could not place, which we call
\textit{leftovers}.  Validation ($R_{16}$) runs once on that
result, after which a profit figure and a verdict are
read off.
Algorithm~\ref{alg:detect} gives the complete
pipeline, writing \textsc{Structural-Fixpoint} for the
reduction of \S\ref{sec:algo:leaves}
and \S\ref{sec:algo:nodes}, which iterates
Algorithm~\ref{alg:leaves} and the node rules of
Table~\ref{tab:labels} until neither locus applies;
Algorithm~\ref{alg:classify} details the
classification logic, with
$\Delta_{\mathrm{gross}}[\tau]{=}\sum_{C}\delta_C[s(C),\tau]$
and $\Delta_{\mathrm{net}}{=}\Delta_{\mathrm{gross}}{-}
\mathit{gas}$, the costs charged against the component
they are denominated in; ``mixed'' = opposite-sign
components.  The two negative conditions of
Algorithm~\ref{alg:classify} differ in scope:
\textsc{neg\_profit} tests that one component,
\textsc{final\_neg} tests them all.

\paragraph{Delta validation}
\textsc{Validate-Deltas} applies $R_{16}$ to every chain
labeled \textit{arbitrage}: it reads
$\delta_C[s(C), \tau_{\mathrm{in}}(C)]$ and downgrades
the label to \textit{cycle} when that entry is not
positive, so that no structurally valid but
unprofitable cycle reaches the profit calculation as an
arbitrage.  It runs here rather than inside annotation
because a chain's delta is complete only once every
sub-chain connected to it has contributed.
When the transaction sender and receiver coincide
(self-call), \textsc{Validate-Deltas} is skipped:
the contract is both origin and intermediary of every
cycle, making per-cycle delta attribution unreliable.

\paragraph{Leftover transfers}
\textsc{Extract-and-Recover} first extracts all
transfers from the reduced AST that were not
incorporated into an arbitrage cycle or a transfer
chain.  These
\textit{leftovers} indicate that the reduction
could not reconstruct a complete flow path for
those transfers.  The two dominant causes are
missing transfer
decodings (non-standard mechanisms invisible to
the decoder) and complex routing that pairwise
chaining cannot resolve.
Leftovers are recorded alongside the detected
cycles; their presence signals that the profit
calculation may be incomplete.

\paragraph{Leftover cycle recovery}
Before computing the final profit, a two-pass
recovery step attempts to extract additional
cycles from the leftovers.

First, it identifies \textit{specular pairs}:
two leftover transfers
$(A, B, a, \tau)$ and $(B, A, a, \tau)$ with
matching token, matching amount, and reversed
endpoints.  These are exact round-trips, characteristic of
flash-loan borrow/repay and other
same-amount refund patterns.  Specular pairs are
removed from the leftover list and grouped into
\textit{leftover cycles}.  Their net delta for
the loan token is zero (the exact borrow amount
is returned); the arbitrageur's profit is
captured separately by the AST cycles that used
the borrowed funds.

Second, remaining leftovers are scanned for
\textit{complementary transfers}: two transfers
involving the same token~$\tau$ where either
endpoint connects
($A = B'$ or $B = A'$), regardless of amount.
These capture partial round-trips where the
repayment differs from the borrowed amount
(e.g., a flash loan with a fee).  Complementary
groups are merged into leftover cycles and their
net delta is added to the accumulator.
Remaining unmatched transfers stay as leftovers
in the final output.

This two-pass recovery is where lending-based
arbitrages are captured: the borrower takes a
flash loan on one side of the call tree and
repays it on the other, too far apart for the
chaining rules to connect.  Specular and
complementary recovery reconstruct these
round-trips from the leftovers.

\paragraph{Profit calculation}
\textsc{Profit} computes the net token balance of
address~$\mathit{to}$ from the cycles and
leftovers produced by
\textsc{Extract-and-Recover}.
Each chain carries the balance $\delta_C$ of
Definition~\ref{def:balance}, maintained as it is
assembled; for a cycle rooted at~$\mathit{to}$, the
entry $\delta_C[\mathit{to}, \tau]$ is the gross profit
in~$\tau$ before costs.  Any remaining leftover
transfer whose destination is~$\mathit{to}$ is
also added to the balance, capturing direct
incoming value that was not part of a cycle.

Three cost components are then deducted:
(i)~the \textit{base fee}
$\mathit{gas\_used} \times \mathit{base\_fee}$
(burned by the network);
(ii)~the \textit{priority fee}
$\mathit{gas\_used} \times
(\mathit{effective\_gas\_price} -
\mathit{base\_fee})$ (tip to the block builder);
and (iii)~an optional \textit{direct builder
payment}
$\mathsf{Leaf}(\mathit{to}, \mathit{builder},
a, \textsc{ETH})$ to the builder's coinbase,
common in MEV strategies.
The first two are read from the transaction's own
metadata, not from the tree, which is why the
native-ETH cost edges of Figure~\ref{fig:txR} never had
to enter~$T_0$.  The third is an ordinary leaf, and the
walk that collects it sets it aside rather than
counting it as profit or as a leftover, so paying the
builder does not by itself make a transaction
inconclusive.
The result is a set of per-token net
balances~$\Delta$ after all costs.

\paragraph{Why profit is attributed at \textit{to}}
The rewriting is indifferent to who profits: it closes
a cycle wherever one closes, and a cycle rooted at an
inner contract is annotated like any other.  The cost
side is what does not generalise.  Gas is paid once, by
the sender, for the transaction as a whole, and no
principled share of it belongs to an inner address that
happens to net positive; without that share there is
nothing to weigh a positive inner balance against, and
a verdict there would report a gross figure as though
it were net.  We therefore attribute only where the
accounting closes, and leave cycles at inner addresses
detected but unclassified.  The choice costs recall and
never soundness, and \S\ref{sec:eval} measures what it
leaves behind.

\subsection{Classification}\label{sec:classify}

Classification inspects~$\Delta$ along sign and
token diversity, yielding four cases:
\textit{positive single-token}
($|\Delta|{=}1$, $a{>}0$), \textit{positive
multi-token} ($|\Delta|{>}1$, all $a_i{>}0$),
\textit{mixed} ($|\Delta|{>}1$ with differing
signs, e.g.\ $+0.5$\,\textsc{weth},
$-200$\,\textsc{usdt}), and \textit{negative}
($|\Delta|{=}1$ with $a{\leq}0$, or
$\Delta{=}\emptyset$).
This classification is computed twice: once on the
gross balance (before costs) and once on the net
balance (after costs).  \textsc{Classify}
(Algorithm~\ref{alg:classify}) turns the six resulting
conditions into a verdict accompanied by an ordered
list of \textit{reasons}, diagnostic labels that
explain the decision to the analyst.  Only four of the
six gate the verdict: a transaction is classified
\textit{Arbitrage} when a cycle survived validation, no
leftovers remain, and the \emph{net} balance is
positive in every component.  When any of the four is
uncertain, the verdict conservatively falls back to
\textit{Warning}, and the gross-balance and net-profit
conditions report without deciding.

\paragraph{Token equivalence}
R14 tests $=_\tau$ (Definition~\ref{def:tokeq}), so a
cycle that leaves in a native asset and returns in its
wrapped form is labeled \textit{arbitrage} and is
covered by Theorem~\ref{thm:sound}.  The mechanization
leaves $=_\tau$ an arbitrary boolean relation on
tokens: no proof invokes any algebraic property of it,
so soundness holds for every instantiation, and the
reflexivity and symmetry a deployment will want are
stated as an obligation on that deployment rather than
assumed anywhere.

For~$\mathit{tx}_R$, the three parallel paths
(Figure~\ref{fig:txR}) reduce to the single arbitrage
cycle of Figure~\ref{fig:ast}(b), spanning all seven
pools, with no leftovers and
$\Delta_{\mathrm{gross}} =
\{(\textsc{weth},\, {+}0.0023)\}$.  Costs leave
$\Delta_{\mathrm{net}} =
\{(\textsc{eth},\, {+}5.5 \times 10^{-5})\}$, reported
in the native asset because the deployment treats it
and its wrapped form as equivalent
(Definition~\ref{def:tokeq}): one component, positive,
so none of the four gates fires and the verdict is
\textit{Arbitrage}.

A verdict of \textit{Arbitrage} certifies closure and
token match from annotation, positivity from
validation, and that every \emph{decoded} transfer is
accounted for; Theorem~\ref{thm:sound} states the
guarantee precisely.

\paragraph{Warnings as diagnostic signals}
A \textit{Warning} means the system cannot
confirm with full confidence, typically because
leftovers indicate incomplete decoding or because
a mixed-sign~$\Delta$ requires cross-token
valuation we deliberately avoid.
The reasons in~$R$ are split into
\textit{verdict-determining}
(\textsc{no\_cycles}, \textsc{leftovers},
\textsc{final\_neg}, \textsc{final\_mixed})
and \textit{explanatory} (all others).
Only the former gate the decision; the latter
provide a forensic diagnostic picture.

\subsection{Formal Properties}\label{sec:algo:formal}

Five statements say what the rewriting may do to a
trace, and therefore what a verdict is worth.  Each is
stated in English first and formally second, proved
against the sixteen rules of Table~\ref{tab:labels},
and assumes Property~\ref{prop:dse} and nothing further
about the execution environment.  Proofs are in
Appendix~\ref{app:proofs}, which opens by separating
what is mechanized from what is not
(Table~\ref{tab:guarantees}).

The reduction re-brackets the transfers it is given.
It drops none, duplicates none, and invents none, so
every balance stays exactly where the trace put it,
whichever order the rules fire in.

\begin{theorem}[Preservation]\label{thm:preserve}
Let $T_0$ be the initial AST and $T_f$ any tree
reachable from it.  Every chain of $T_f$ tracks walks
in $G$ (Definition~\ref{def:tg}); every transfer of
$T_f$ is a transfer of $T_0$; and $T_0$ and $T_f$ carry
the same multiset of transfers.
\end{theorem}

The reduction also stops, and stops soon.  Two counters
fall as it runs: the chains still unlabelled, $u$, and
the total number of children, $c$.  In the
annotation-and-connection loop neither rises and every
non-trivial pass consumes a unit of their sum, which
on a fully lifted non-empty input takes the arithmetic
form $3n - 2$, the bound on passes.  The unrestricted
relation ranks trees by $(c, u)$ instead, which every
rule decreases (Appendix~\ref{app:proofs}).  Lifting
is not needed to terminate, only to give the budget
its arithmetic form.

\begin{theorem}[Termination]\label{thm:term}
The deterministic annotation-and-connection fixpoint
reaches a normal form from every reduced AST.  If its
input is fully lifted and non-empty with $n$
transfers, it takes at most $3n - 2$ passes.  The full
rewriting relation is well-founded under arbitrary
rule orders.
\end{theorem}

Soundness is the statement the rest of the paper is
written to support, so it is worth being exact about
what it does and does not say.  It concludes about the
transaction, not about our labels: when the pipeline
answers \textit{Arbitrage}, the transaction's own
transfer graph contains an arbitrage in the sense of
Definition~\ref{def:arb}, a definition that names
neither the rewriting, nor the normal form, nor any
label the reduction assigns.

\begin{theorem}[Soundness]\label{thm:sound}
Let $T_0$ be a decoded trace and $T_f$ any tree
reachable from it.  If \textsc{Classify} returns
\textit{Arbitrage} on $T_f$ and at least one cycle
survives $R_{16}$ and the net-positive gate, then
there are an address $v$, a
token $\tau$ and a bundle of walks over the transfers
of $T_0$ that together form an arbitrage at $v$ in
$\tau$ (Definition~\ref{def:arb}).
\end{theorem}

Its two structural hypotheses hold of any decoded
trace.  The third is validation itself, stronger than
Definition~\ref{def:arb} needs
(Appendix~\ref{app:proofs}); the fixpoint path
discharges it by construction, since validation runs
before the cycles are read off.  The
one path that can report \textit{Arbitrage} without
a surviving cycle is the post-fixpoint promotion of
\S\ref{sec:eval}: $0.8\%$ of detections, the only
component outside the mechanized core.

\begin{proof}[Sketch]
Three steps, of which only the second has content.
(i)~Reading the verdict backwards is mere inversion:
\textit{Arbitrage} requires a chain of $T_f$ carrying
the \textsc{arb} label that $R_{16}$ did not revoke.
(ii)~By Theorem~\ref{thm:preserve} the reduction can
neither manufacture nor duplicate a transfer, so that
chain's transfers are transfers of $T_0$, counted with
multiplicity.  This is the bridge, and what makes the
conclusion a claim about the transaction rather than a
restatement of how we labelled it.
(iii)~Closure, token match and a positive balance at
the origin are read off the chain and transfer
verbatim to $G$, and its maximal decomposition into
connected runs is the bundle
Definition~\ref{def:arb} asks for.
\end{proof}

Two consequences follow.  Every balance the system
reports is a balance of the input trace, so a profit
figure cannot be an artefact of the reduction.  And
Definition~\ref{def:arb} is not satisfied by anything
merely circular: a two-leg swap leaving in one asset
and returning in another meets every other clause and
is rejected on the round-trip condition alone, which
we check as a mechanized counterexample.

Convergence is where the execution model earns its
keep.  Left free, the rules can be applied in many
orders; the kernel fixes one, scanning the EVM's trace
order left to right, and under it the answer is unique.

\begin{theorem}[Uniqueness]\label{thm:confluence}
Under the deterministic kernel over the execution's
trace order, every AST has a unique normal form.  Under \emph{every} order,
any two reachable trees agree on the balance
$\delta[v, \tau]$ at every address $v$ and token
$\tau$.  On well-formed trees, two one-step
divergences are joinable up to reassociation of
parallel merges, the same branches grouped
differently.
\end{theorem}

The three clauses answer three different questions, and
only the first is about syntax.  Well-formedness is
linearity, pairwise-distinct trace positions, which
decoded trees have by construction and every rule
preserves, plus a frame condition at the divergence
site (Appendix~\ref{app:proofs}).  Whether local joinability extends to a global
statement is work in progress; every result in this
paper uses the deterministic order and is unaffected
either way.

\begin{theorem}[Decidability]\label{thm:decidable}
Under the deterministic kernel, two ASTs are joinable
if and only if their kernel normal forms are equal,
and equality of those normal forms is decidable.  The
coarser equivalence that identifies trees differing
only by reassociation of parallel merges is decidable
as well.
\end{theorem}

Preservation and termination follow the two counters
above by induction on rule applications.  Uniqueness
follows because the step function is total, so at
most one rewrite fires per tree.  Decidability
follows from uniqueness, a terminating computation of
normal forms, and structural comparison.  For a terminating confluent
rewrite system the word problem is classically
decidable~\cite{huet1980,bookotto1993,baadernipkow1998,terese2003};
here it follows instead from termination, the
deterministic kernel, and structural comparison of
normal forms.

\paragraph{What is certified, and what is trusted}
The five theorems are mechanized in Rocq: $13\,501$
lines, 432 lemmas, 13 corollaries and one worked
counterexample, with no admitted obligation and no
added axiom.  The development declares eleven
parameters, among them the address and token types,
token equivalence, the burn and mint predicates, and
the router hint of R5.  They are declared and never
defined, and no proof assumes anything about them
beyond their types, so the theorems hold for every
instantiation a deployment chooses.  Two
well-formedness conditions are stated but deliberately
never assumed, the reflexivity and symmetry of $=_\tau$
among them; they are obligations on the deployment, not
hypotheses of the results.  Appendix~\ref{app:proofs}
maps each theorem to its mechanized statement.

\paragraph{Specification and implementation}
The system in production is OCaml, and the mechanized
rules mirror it one for one: each rule of
Table~\ref{tab:labels} is a constructor, and the
mechanized kernel implements the same priority policy
as the production module.  That \emph{kernel} is
extractable, and the extracted module compiles and
runs with placeholder realizers for the declared
parameters, so the specification is executable rather
than merely readable.  What we have not done is
differential testing of the two at trace scale; we
commit to it, and until then their correspondence is a
matter of construction and review, not measurement.
Property~\ref{prop:dse} is not assumed anywhere: it is
encoded in the data structure, since the tree of
Definition~\ref{def:cft} has ordered children, which is
exactly what the property guarantees.

\subsection{Implications of Decidability}
\label{sec:algo:implications}

Theorem~\ref{thm:decidable} makes the normal form a
substrate rather than an answer: any decidable
predicate over normal forms is a query, and detection
is only the first of them.  Exact structural
equivalence is the second, deciding whether two
transactions reduce to the same form; a coarser
\emph{family} equivalence is the third, composing
that test with a skeleton map $\mathsf{sk}$ that
erases path lengths, arities and token identities,
so a three-arm and a five-arm star fall together.
Each composes normalization with a decidable test
on trees, so each is decidable and linear in the
size of the normal form:
under $\mathsf{sk}$, $\mathit{tx}_R$ becomes
$\mathsf{Star}(\star,\, \mathsf{Swap}(\star,\star),\,
\star)$, as does any multi-arm star with two-hop arms,
whatever its arity or tokens.
Appendix~\ref{sec:equivalence} develops the query
language and a context-free grammar schema whose
languages are the strategy families, and
Appendix~\ref{sec:grammar:flashloan} extends the
grammar to flash loans.

\section{Evaluation}\label{sec:eval}

The theorems say what a verdict certifies; this section
measures what the system finds.  We evaluate the
production OCaml module, deployed as the analysis layer
of a blockchain forensic platform, against EigenPhi on
220\,000 Ethereum blocks and against ArbiNet on a
shared 1\,000-block subset.

\subsection{Setup and Baselines}\label{sec:eval:dataset}

We run the system on an arbitrary contiguous range
of 220\,000 Ethereum mainnet blocks
(23\,699\,751 through 23\,919\,750,
31~October--1~December 2025), taken within the
period for which EigenPhi labels were available; a
separate arbitrary 1\,000-block range
(24\,100\,000--24\,100\,999) is used for the
three-way comparison.  Neither range was
selected with any knowledge of what it contained,
and every transaction in both is analysed.
For each block, we obtain the full execution trace
of every transaction via
\texttt{debug\_\allowbreak traceTransaction} on an
archive node.  The decode layer extracts transfers
using the standard ERC token ABIs and WETH,
a fixed set shared across all token contracts.  The
detection algorithm then receives only typed
transfer tuples and the call hierarchy; it has no
knowledge of Uniswap, Curve, Aave, or any other
protocol.  Under these conditions the system flags
790\,593 transactions, across the four tiers defined
below, as exhibiting arbitrage cycle structure.  We release the per-transaction verdicts
(hash, block, verdict, reasons, cycle count, and
per-transaction latency) as a public dataset,
enabling independent reproduction of the
evaluation without archive-node access.

\paragraph{Baselines}
No public, independently annotated ground truth
exists for per-transaction arbitrage detection;
prior work relies on self-collected
labels~\cite{defirangertdsc,chi2024} or on
third-party MEV data~\cite{qin2022quantifying}.
We therefore compare against the only two systems
in this domain whose labels or code we could
reuse: \textit{EigenPhi}~\cite{eigenphi}, a
production MEV platform that flags 649\,790
transactions in the same range, and
\textit{ArbiNet}~\cite{arbinet2024}, a GNN
classifier that, like our analysis layer, needs no
protocol-specific ABI catalog.
The EigenPhi comparison is deliberately asymmetric,
since EigenPhi leverages protocol knowledge our
system lacks, so we frame all comparisons as
inter-system agreement and analyse the
disagreements in both directions.

\paragraph{Two vocabularies}
Two orthogonal partitions organize what follows.
\textit{Tiers} grade our verdicts by the reasons
of Algorithm~\ref{alg:classify}:
\textit{confirmed} (clean cycles, positive net, no
leftovers), \textit{probable} (leftovers remain),
\textit{attempted} (cycles whose gas exceeds the
profit), and \textit{uncertain} (mixed-sign
balances).  \textit{Categories} classify
agreement with EigenPhi: both flag
(Category~1), only we flag, split by tier into
confirmed (Category~2) and warning (Category~3);
only EigenPhi flags (Category~4).  Tiers are properties of
one system's output; categories compare two.

Two scope choices inherited from deployment affect
the numbers below.  Profit is attributed only at
the transaction recipient, for the cost-attribution
reason of \S\ref{sec:algo:leftover}, so cycles at
inner contract addresses are detected but not
classified; the disagreement analyses measure what
this leaves behind.  And the
post-fixpoint promotion path scoped by
Theorem~\ref{thm:sound} accounts for 0.8\% of
detections (6\,147).  No sampled transaction in
\S\ref{sec:eval:manual} came from it, so the
statistical validation does not cover it.  We
uniformly sampled 10 promoted detections and
inspected their raw traces; all 10 were manually
adjudicated as arbitrages.  Otherwise the case rests on the
firing conditions,
$\tau_{\mathrm{in}} =_\tau \tau_{\mathrm{out}}$ with
positive delta and no leftovers.

\subsection{Agreement with EigenPhi}\label{sec:eval:accuracy}

The diagnostic reasons partition the 790\,593
detections into the four tiers, each with a
distinct EigenPhi agreement rate.
The \textit{confirmed} tier (469\,801 detections,
59.4\%) is independently flagged by EigenPhi at
87.2\%; \textit{probable} (5.0\%) at 70.8\%;
\textit{attempted} (245\,497, 31.1\%) at 39.5\%,
despite the negative net; \textit{uncertain}
(4.5\%) at 21.0\%.  This monotonic gradient was not
engineered: it emerged from
Algorithm~\ref{alg:classify} and is an independent
sanity check, in that the more confident our
reasons, the more often the baseline agrees.
Overall agreement is 83.5\% (542\,279 of
649\,790; EigenPhi covers profitable and
unprofitable cycles), splitting 75.5/5.2/17.9/1.4
across our tiers.  The remaining 107\,511
EigenPhi-only transactions are analysed in
\S\ref{sec:eval:cat4}.

In the other direction, 248\,314 transactions are
flagged by our system and not by EigenPhi:
60\,199 confirmed (12.8\% of all detections),
148\,422 attempted, 11\,612 probable, and 28\,081
uncertain.  This volume, invisible to a profit-based
detector, is consistent with Li et
al.~\cite{li2023actlifter} and Chi et al.~\cite{chi2024}.

\subsection{Manual Validation}\label{sec:eval:manual}

We complement the large-scale comparison with
transaction-level inspection by the authors: we
uniformly sample 100~transactions from each of
Categories 1--3 (stratified by tier), resolve
addresses via Blockscout, and verify each cycle
against the raw trace, the tree reductions, and
the economic intent.  Every verdict is
\emph{structurally verifiable}: the reduced AST,
transfer chains, and delta map accompanying each
detection constitute an independently checkable
certificate, so an auditor can confirm a verdict
by inspecting the reduced tree alone, without
re-running the reduction
(Appendix~\ref{sec:casestudy} demonstrates this
on three forensic examples).

\paragraph{Category~1: both confirmed (100~txs)}
All sampled transactions are genuine arbitrages
through identifiable DEX pools (7 AMM protocols
represented, led by Uniswap~V3/V2/V4 and Curve).
91\% are single-cycle arbitrages, 9\% multi-cycle
batches; cycles span 2 to 31 transfer leaves,
median 5.  2\% involve flash loans.
\textbf{Zero false positives.}

\paragraph{Category~2: ours-only confirmed (100~txs)}
All 100 sampled exclusive detections are genuine
arbitrages.  They are larger and more
flash-loan-heavy than Category~1: cycle hops
median 8 (vs 5), and 51\% involve flash loans
(vs.\ 2\%), detected without lending-specific
logic via leftover recovery's structural pairing.
Appendix~\ref{sec:casestudy:fp} walks through one
such case: a three-hop flash-loan arbitrage that
reduces to two closed WETH cycles
(17{,}912~WETH profit, EigenPhi-missed).
\textbf{Zero false positives.}

Across the 200 confirmed-tier samples of
Categories 1 and~2, manual review finds no false
positive; a sample of this size bounds the true
rate below 1.5\% with 95\%
confidence.\footnote{Clopper-Pearson one-sided
upper limit for 0 of 200: 1.49\%.}
All 200 confirmed-tier samples in Categories~1 and~2
came from the fixpoint, so this bound covers the
mechanized core and not the 0.8\% promotion path.

\paragraph{Category~3: ours-only warnings (100~txs)}
All sampled warning-tier transactions exhibit
structurally valid cycles through identifiable DEX
pools.  Manual review labels 78\% attempted
arbitrages, 13\% uncertain, and 9\%
probable.  21\% are multi-cycle batches (vs 9\%
in Category~1), 18\% involve flash
loans, and 38\% settle through CoW Protocol
(intent-based solver flows whose cyclic structure
does not reduce to single-actor arbitrage).  The
warning tier surfaces these without committing to
an arbitrage verdict.

\subsection{Where We Disagree, and Who Adjudicates}
\label{sec:eval:cat4}

Category~4 is the opposite question: what backs
the 107\,511 transactions EigenPhi flags and we do
not?  We sample 200 (a larger sample, for tighter
forensic statistics), rerun them with full
diagnostic output, and inspect the resulting
normal forms.  All 200
receive a non-arbitrage classification, in three
groups.  127 (63.5\%) exhibit no arbitrage cycle
in normal form: simple transfers, wrap/unwrap
operations, RWA-token mints, or yield harvests
(Appendix~\ref{sec:casestudy:tn} inspects one
directly).
A further 55 (27.5\%) contain structural cycles
($s(C) = d(C)$) whose entry and exit tokens
differ: cross-token routing paths, not arbitrages.
The remaining 18 (9.0\%) contain genuine arbitrage
cycles at an inner contract address (flash loan
executor, router delegate), not at the transaction
recipient: the classification declines to surface
them.  We adjudicate these 18 as true misses.

This adjudication has a circularity we do not wish
to hide: the instrument that classifies the
disagreement is our own normal form.  The
first group is the one a skeptic should press on,
and it admits an independent check that never
invokes the rewriting: enumerate cycles directly
in the decoded transfer graph with a textbook
algorithm~\cite{johnson1975}, test closure and
token equivalence per Definition~\ref{def:arb},
and compare the verdicts.  We commit to this
oracle in the artifact; on the cases already
inspected by hand (Appendix~\ref{sec:casestudy}),
the trace itself settles the question, since no
cycle exists to find.

\subsection{Three-Way Comparison}\label{sec:eval:arbinet}

ArbiNet was trained on blocks
15\,540\,000--15\,585\,000 (September~2022); we run
the pretrained model unmodified on the 1\,000-block
range.  The three-year gap stress-tests temporal
generalization, which a structural detector does not
need, though the model's age may disadvantage
ArbiNet.  Table~\ref{tab:threeway} reports the eight
cells.  Of 3\,193 transactions flagged by at least
one system, all three agree on 1\,158; our system
covers 81\% of ArbiNet's detections and 92\% of
EigenPhi's.  634 are exclusive to our system: 148
confirmed, 411 attempted, 34 probable, 41 uncertain.
For the 432 flagged by a baseline and not classified
by us, the gap analysis
(rightmost column) repeats the Category-4
methodology: ArbiNet-exclusive detections split into
39.1\% inner-address cycles, 45.7\% cross-token, and
15.2\% with no cyclic structure; EigenPhi-exclusive
into 24.7\%, 72.7\%, and 2.6\%.

\begin{table}[t]
\caption{Three-way comparison on 1\,000 blocks.
  Eig. = EigenPhi; inner = arbitrage cycle at an
  inner address; x-tok = cross-token; NC = no cycle
  in normal form.}
\label{tab:threeway}
\centering
\footnotesize
\setlength{\tabcolsep}{3pt}
\begin{tabular}{@{}lccl@{}}
\toprule
& Eig.+ & Eig.$-$ & Gap analysis \\
\midrule
Ours+ ArbiNet+ & 1{,}158 & 342 & --- \\
Ours+ ArbiNet$-$ & 627 & 634 & --- \\
Ours$-$ ArbiNet+ & 74 & 279
  & 39.1\% inner, 45.7\% x-tok, 15.2\% NC \\
Ours$-$ ArbiNet$-$ & 79 & ---
  & 24.7\% inner, 72.7\% x-tok, 2.6\% NC \\
\bottomrule
\end{tabular}
\end{table}

A per-block case study
(Appendix~\ref{sec:arbinet:case}) reveals two
qualitative differences.
ArbiNet produces false positives on batch
settlement protocols (CoW Protocol): the GNN sees
cyclic graph topology but does not verify that the
tokens match at cycle boundaries
(Definition~\ref{def:arb}).  Conversely, ArbiNet
misses arbitrages that both our system and
EigenPhi confirm, in transactions whose routing
contracts postdate its 2022 training data.  A
structural detector requires no retraining because
it operates on transfer events, not on learned
representations of specific contracts.

\subsection{Topology, and a Second Query}
\label{sec:eval:queries}

The reduction finds 940\,760 individual cycles
across the 790\,593 flagged transactions.  After
structural merging, 89.1\% of transactions reduce
to a single cycle, 6.7\% to two, and 4.1\% to
three or more (max 82); parallel arbitrages
sharing an origin merge into one cycle via
R7--R9.  Within each cycle, even lengths dominate
(an AMM swap produces two transfers), and odd
lengths arise from single-token relays.  28.5\% of
warnings contain leftover cycles recovered by the
two-pass procedure of \S\ref{sec:algo:leftover},
and 51\% of exclusive confirmed detections involve
flash-loan round-trips discovered by its
structural pairing.
The equivalence queries also run on this data: a
second confirmed detection~$\mathit{tx}_S$, built
from different tokens through twice the paths,
reduces to the same skeleton as $\mathit{tx}_R$
(Appendix~\ref{sec:equiv:example}).

\begin{figure}[t]
\centering
\includegraphics[width=\columnwidth]%
  {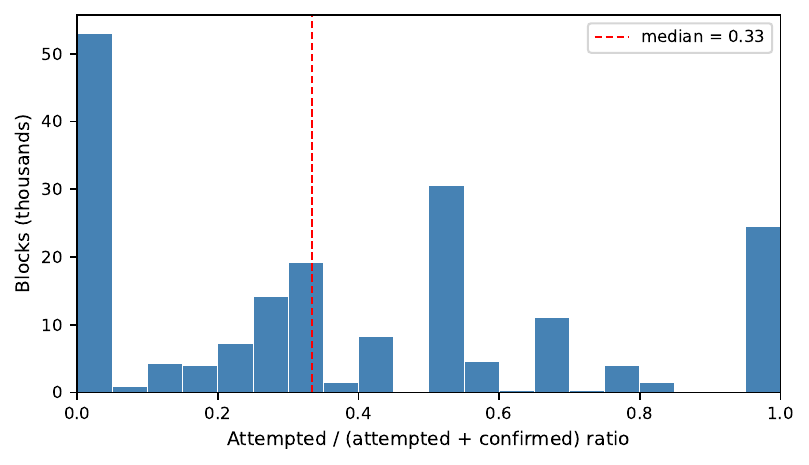}
\caption{Per-block attempted / total arbitrage
  ratio (median = 0.33).}
\label{fig:attempted}
\end{figure}

\paragraph{Attempted arbitrages}
The \textit{attempted} tier is a further query in
the sense of \S\ref{sec:algo:implications}: the same
normal form, asked a different question.  Closure
and token match are kept, profit positivity is
dropped, and what answers is the population of
cycles that closed structurally but lost money
after gas: 245\,497 transactions, roughly half the
confirmed count, consistent with a bot that built
the same multi-pool cycle as a profitable arbitrage
but lost the race, overpaid gas, or encountered
unfavourable slippage.  Profit-based detectors
discard them by construction.
At the per-block level, the attempted-to-total
ratio clusters at low-denominator fractions
(Figure~\ref{fig:attempted}): 27.9\% of blocks at
ratio~0 (only profitable cycles), 16.0\% at $1/2$,
and 12.9\% at 1, blocks whose every cycle lost.
The pattern fits small-$k$ competition (a handful
of cycles per block) rather than smooth many-bot
racing.  The ratio-0 mode reflects private
orderflow or narrow opportunity
windows~\cite{qin2022quantifying}; the ratio-1
mode is not predicted by the gas-auction model of
Daian~et~al.~\cite{daian2020flashboys}, which
assumes a winner, and is consistent with capture
via sandwich or off-chain
routing~\cite{chi2024,heimbach2024nonatomic,capponi2025mev}.
Profit-based detectors see only the winners;
structural detection sees both sides.

\subsection{Performance}\label{sec:eval:perf}

The system handles arbitrary token counts: observed
inputs range from 3 distinct tokens at the median
(31.8\% triangular) to 106 in the tail.
Theorem~\ref{thm:term} bounds \emph{passes} at
$3n - 2$ and says nothing about the cost of one.
As implemented, a pass looks for a partner for each
child by walking its remaining siblings, so a pass
is quadratic in the worst case and the guaranteed
total cubic; in practice every successful rewrite
shrinks the child list, and the latencies below
stay far from the bound.  An index keyed on
junction address and token would make the lookup
constant-time and the total quadratic; the
measurements have not asked for it.
The rewriting runs in 0.07\,ms median
and 0.47\,ms P95, under a third of per-transaction
cost; the rest is decoding.
Including decoding but excluding trace retrieval,
the full pipeline completes in 0.25\,ms (median),
2.24\,ms (P95), and 6.15\,ms (P99).
0.02\% of transactions containing an arbitrage
cycle exceed 100\,ms; these are high-transfer
transactions (median 181 transfers, max 2008),
where the per-pass sibling walk dominates.

\section{Related Work}\label{sec:relw}

The study of value extraction on Ethereum began with
Daian et al.~\cite{daian2020flashboys}, who coined
\textit{miner extractable value} (MEV) and documented
the priority-gas auctions among arbitrage bots.
Qin et al.~\cite{qin2022quantifying} provided an
early systematic measurement (\$540M over 32~months),
but acknowledged it as a lower bound: only
transactions matching known protocol signatures were
counted. The SoK of Zhou et al.~\cite{zhou2023sokdefi}
generalizes this gap beyond arbitrage: across 181
DeFi incidents (\$3.24B), cross-protocol composability
is the axis existing tools consistently miss.

A first wave of detectors builds aggregate token
graphs across blocks.
McLaughlin et al.~\cite{mclaughlin2023} apply
Johnson's algorithm to \texttt{Transfer} events
and uncover 3.8M arbitrages.
Zhang et al.~\cite{zhang2024} search the pool
graph itself with a Bellman-Ford variant, finding
opportunities rather than executions, and
Wang et al.~\cite{wang2022cyclic} characterize
cyclic arbitrages across DEXs.

These methods gain scale but lose the execution
context of individual transactions: a forensic
analyst cannot reconstruct \textit{how} a profit
was realized from an aggregate cycle count.
More fundamentally, cycles in the flat transfer
graph do not reliably indicate arbitrage: a yield
harvest that withdraws and redeposits the same
token produces a cyclic flow with matching tokens
and positive balance, yet involves no price
exploitation (Appendix~\ref{sec:casestudy:tn}).
Distinguishing such patterns requires the call
hierarchy that only per-trace analysis preserves,
refining these cycles
(Theorem~\ref{thm:decidable}).

DeFiRanger~\cite{defirangertdsc} and
ActLifter~\cite{li2023actlifter} identify DeFi
actions via protocol-specific event catalogs.
DeFiRanger uses a fixed pattern catalog for price
manipulation.  ActLifter curates 88~events across
10 DeFi actions and clusters bundles to discover
new MEV strategies, complementing our
per-transaction scope.
Chi et al.~\cite{chi2024} reduce fragility with
exchange-rate checks but still need 44~hand-crafted
patterns.
None of these provide formal guarantees, and each
requires per-protocol catalog extension.
Our rewriting system normalizes structurally: one
rule covers all protocols sharing a pattern,
recovering multi-cycle split-and-merge
and lending-based
arbitrages without manual extension.
Richer trace representations~\cite{txspector2020,qin2025clue,etrace2025}
target vulnerability classification, not arbitrage.

Catalog-free machine-learning approaches pursue
generality through learned representations rather
than algebraic structure.
Park et al.~\cite{arbinet2024} propose ArbiNet,
a GNN-based classifier that, like our analysis
layer, requires no protocol-specific ABI catalog.
ArbiNet produces binary labels from training data.
Our algorithm reconstructs transfer chains with
profit calculations and diagnostic reasons, and
offers mechanized formal guarantees.
Section~\ref{sec:eval:arbinet} provides a
quantitative comparison on 1\,000 shared blocks.
Jin et al.~\cite{jin2022} and
Niedermayer et al.~\cite{niedermayer2024} classify
arbitrage \textit{addresses} and financial bots,
answering \textit{who} performs arbitrage but not
\textit{how} a given transaction realizes its
profit.

Formal-methods approaches split a priori from
a posteriori: Clockwork
Finance~\cite{babel2023clockwork} enumerates
profitable attacks from protocol specifications, while
the Imitation Game~\cite{qin2023imitation} replays
observed strategies.  Our normal forms
(Section~\ref{sec:algo:implications}) partition
realized transactions. A formal bridge to
Clockwork's symbolic state space is open future
work.
KEVM~\cite{kevm2018} formalizes EVM bytecode
semantics. We abstract above bytecode, on the
typed-transfer view the decoder produces.

MEV strategies compose across
transactions~\cite{linkingmev2025,torres2021frontrunner}
and across
chains~\cite{sokmev2026,torres2024rolling}.
Cao et al.~\cite{cao2026peb} show that
transfer-graph abstractions are structurally
incomplete for beneficiary attribution and
advocate execution-semantic analysis. Our approach
already operates at that level.

Materwala et al.~\cite{mevsurvey2024} note that existing detectors
trade off protocol coverage, per-transaction granularity,
and explainability. Ours makes arbitrage one query on a
normal-form substrate with mechanized structural soundness
and decidable equivalence.

\section{Limitations}\label{sec:limits}

\paragraph{Trusted base and threat model}
We assume an EVM-compatible chain whose traces
faithfully reflect consensus semantics: tracing,
client consistency, execution visibility and
consensus form the trusted base.
The theorems apply to the
transfers the decoder emits, and an adversary acts
on that boundary.  The algorithm ignores
protocol-specific events, so renaming or omitting
them, proxying, or obfuscating control flow gains
nothing.  \emph{Suppressing} the standard
value-movement events, bypassing them through
internal balances, or \emph{fabricating} events
that move no value: these all split by where the
hidden value sits.  A transfer missing from a cycle
breaks the chain and costs recall, while a hidden cost outside
the cycle leaves the structural verdict standing
and overstates the reported net.  Verdicts are
sound about decoded transfers, not about value;
cross-checking them against the transaction's state
diff would discharge event fidelity, a check we
commit to alongside the differential testing of
\S\ref{sec:algo:formal}.

\paragraph{Scope}
Completeness is not claimed: missing or
netted-away transfers silently cost recall, so the
recall figures of \S\ref{sec:eval} are measured,
not guaranteed.  The system reads one
transaction's trace and cannot see strategies
spanning multiple transactions, bridges, or
non-atomic flows, which Heimbach et
al.~\cite{heimbach2024nonatomic} place above 25\%
of DEX volume; lifting to block level is the
natural next step, as a block is a sequence of
normal forms and no per-transaction theorem
changes.  The evaluation covers a single 30-day
window; no public historical MEV label sets exist.

\section{Conclusion}\label{sec:conc}

The rewriting system is not a detector with formal
properties: it is an algebra whose simplest query
is detection.  Every transaction admits a normal
form under the deterministic kernel over its
execution-fixed trace order, and the structural
equivalence this induces is decidable
(Theorem~\ref{thm:decidable}), with arbitrage
detection as its witness.  Cycles emerge from the
fixpoint over the call hierarchy and the sender
field~$\sigma$ rather than being searched for.  The
five properties are mechanized in Rocq with no
admitted obligation, so the 99.2\% of detections
that come from the fixpoint are sound over the
decoded transfers, and the normal form makes every
verdict auditable.  On 220\,000 blocks it matches a
production platform's coverage while surfacing
60\,199 exclusive detections and 245\,497 attempted
arbitrages, a tier profit-based detectors discard.

The same binary runs on Arbitrum and BNB Smart
Chain with a configuration
change (Appendix~\ref{sec:arbitrum}).  New
annotation predicates should expose liquidations,
oracle manipulation, and, with block-level
composition, sandwich attacks.  Equivalence classes
of normal forms partition fund flows into a
taxonomy with no manual labeling.  The normal forms
also suggest an abstract domain for fund flows and a
route toward gas-bounded abstract interpretation of
smart contracts into that domain.

For now: if it walks like an arbitrage, the normal
form confirms it.  If it wraps like a sandwich,
burns like a liquidation, or \textit{\ldots}, the
normal form will carry the answer in its shape.

\section*{Ethics Considerations}

\paragraph{Data}
Everything we analyze is public: transactions,
execution traces and transfer events recorded on a
permissionless blockchain and readable by anyone with
an archive node.  The system only reads, and every
measurement replays history that had already been
confirmed, so the experiments cannot have influenced
the market they measure.  No IRB review was required.
That exemption is not the interesting question,
though.  Addresses belong to people and programs
those people wrote, and we publish judgements about
what they did, so the risks worth discussing are
about the judgements.

\paragraph{Dual-use risk}
The detection system produces, for each flagged
transaction, the complete transfer chains, profit
calculation, and pool identifiers.  In principle,
this output could help a searcher reverse-engineer
a competitor's arbitrage strategy.  We note three
mitigating factors.
First, the system operates on \emph{confirmed}
transactions: the opportunities it describes have
already been executed and are no longer available.
Second, the same information is already accessible
through commercial platforms (EigenPhi, Flashbots
Transparency Dashboard) and on-chain explorers
(Etherscan, Blockscout).  Our contribution is a
\emph{structural representation}, not new data
access.
Third, operational MEV extraction requires mempool
monitoring, gas-price optimization, and
sub-second execution infrastructure, none of which
the paper provides or discusses.

\paragraph{Privacy and de-anonymization}
The decode layer enriches addresses with OSINT
metadata (contract names, token labels, entity tags)
obtained from public sources (Sourcify, 4byte
directory, Blockscout).  The \emph{analysis layer
does not use} these labels: detection operates
exclusively on transfer tuples and the call
hierarchy (Section~\ref{sec:arch}).  The published
dataset contains only transaction hashes, block
numbers, structural verdicts, and timing
measurements.  No OSINT labels, wallet balances, or
entity identifiers are included.
Transactions appear by hash, as they do on chain; we
name no person or organization and do not try to link
an address to anyone.
Verdicts correlated with address activity could still
help de-anonymize a pseudonymous actor, given other
data.  That is true of any public-chain analysis and
not special to us, but it did shape what we release:
the dataset is per transaction and aggregates nothing
by address, so it describes trades rather than
traders.

\paragraph{What a verdict means, and what it does not}
The system answers one structural question: did value
leave an address and come back, in the same asset and
in greater quantity.  Intent, legality and harm are
different questions, and the algorithm has no access
to them.  The gap matters in both directions.  A
confirmed arbitrage is often routine market-making
that corrects a price difference between pools.  And
a cycle can close through an address that is not a
trading venue at all, in which case the shape is an
accident of how the value moved rather than a strategy
anyone executed.  We saw both while building the
evaluation.

Three properties of the design keep that gap from
becoming a false accusation.  Soundness is claimed for
the structural property alone, over the transaction's
own transfers (Theorem~\ref{thm:sound}).  The tiers
abstain: a reconstruction with leftovers or an
ambiguous balance returns \textit{Warning} instead of
a verdict, which is why only 59.4\% of detections are
called arbitrages at all.  And every verdict carries the
reduced tree, the chains and the balances that
produced it, so it can be checked rather than
believed.  A verdict is a place to start looking, not
a conclusion about a person, and reading it as the
latter, in an enforcement setting and without a human
in the loop, is the misuse we would warn against.

\paragraph{The legal cases we cite}
Section~\ref{sec:intro} opens on live proceedings
because they show what turns on the question, not
because we analyze them.  We describe them from public
records, analyze no transaction connected to any of
them, and take no view on any party.  The system has
not been validated for evidentiary use and we do not
offer it as such.

\bibliographystyle{IEEEtran}
\bibliography{references}

\appendices

\section{Generative AI Usage}\label{app:ai}
Anthropic Claude (Opus~4.8, Opus~5, Fable~5) was used as a
writing and engineering assistant: to edit and condense
author-written text, to check the paper for internal
consistency and against the artifact, and, in the Rocq
development, for tactic scripts and refactoring under
the authors' direction, where every obligation is
machine-checked.  It also drafted the artifact's
documentation files, and raised and drafted much of
the Ethics Considerations section, which the authors
then revised and adopted as their own.
All content, structure, claims,
theorem statements, proof strategies, the algorithm and
its implementation, the evaluation design, and every
experimental number are the authors' own work, and the
authors take responsibility for all material in the
paper.

\section{Open Science}\label{app:openscience}
We release the following artifacts, sufficient
to evaluate every core contribution of the paper.
Item~1 reproduces the formal proofs;
items~2--5 reproduce the empirical evaluation.

\begin{enumerate}
\item \textbf{Rocq formalization.}
  The file \texttt{Arbitrage.v}\footnote{\url{https://anonymous.4open.science/r/artifact-submission-2128/rocq/Arbitrage.v}}
  (5~theorems, 432~lemmas, 13~corollaries,
  13\,501~lines, no admitted obligation, no added
  axiom, eleven declared parameters) is
  self-contained and compiles with
  \texttt{rocq compile Arbitrage.v}.
  It verifies all five theorems.
  \texttt{rewrite\_step} has 17 constructors: the 15
  rules R1--R15 of Table~\ref{tab:labels}, plus
  \texttt{RS\_lift} and the congruence
  \texttt{RS\_under}.  R16 is modeled by
  \texttt{validated\_arbitrage}, which conjoins its
  gross-delta test with the deployment's cost model
  \texttt{net\_positive}, so the mechanized premise
  bundles R16 with the net-profit gate that
  \S\ref{sec:classify} applies afterwards; a cycle
  passing the first and failing the second is the
  \textit{attempted} tier
  (\texttt{attempted\_arbitrage\_characterization}).
  The step
  function is computable, making it amenable to
  Rocq's \texttt{Extraction} mechanism.
  A long-form proof companion\footnote{\url{https://anonymous.4open.science/r/artifact-submission-2128/proofs/proofs.pdf}}
  expands Appendix~\ref{app:proofs} and develops the
  structural-equivalence material of
  Appendix~\ref{sec:equivalence} at greater length.

\item \textbf{Evaluation dataset.}
  The evaluation rests on three data files.
  (i)~A per-transaction CSV recording the verdict,
  diagnostic reasons, cycle count, and latency for
  every flagged transaction across the 220K-block
  range.  (ii)~The EigenPhi arbitrage labels for
  the same range, enabling the cross-reference
  analysis.  (iii)~The ArbiNet predictions on the
  shared 1K-block range, enabling the three-way
  comparison.
  Together, these files allow the evaluation
  pipeline to reproduce every statistic, figure,
  and table in Section~\ref{sec:eval} without
  archive-node access.

\item \textbf{Evaluation pipeline.}
  25~Python scripts orchestrated by
  \texttt{run\_all.py} and documented in
  \texttt{METHODOLOGY.md}.  Given the evaluation
  dataset above, the pipeline produces all
  figures, tables, and summary statistics.
  The scripts also automate the manual-validation
  sampling, the Category~4 forensic analysis,
  and the three-way gap analysis.

\item \textbf{Detection tool.}
  A Docker image containing five compiled
  binaries (no source code):
  \texttt{inspect\_tx\_offline} and
  \texttt{batch\_analyse\_offline} reproduce
  detection verdicts from pre-exported traces
  (single-transaction and block-range CSV
  outputs respectively);
  \texttt{inspect\_tx} and
  \texttt{batch\_analyse} perform the same two
  roles against any EVM-compatible RPC endpoint
  (reviewers can verify any transaction in the
  paper or test on Arbitrum, BSC, or other EVM
  chains by changing the RPC URL and wrapped-token
  address);
  \texttt{export\_data} exports raw traces from
  an RPC node for subsequent offline analysis.

\item \textbf{Pre-exported traces.}
  The full 220K-block evaluation produces over
  500\,GB of raw execution traces, which exceeds
  practical artifact hosting limits.  We therefore
  provide raw traces for 1\,000~blocks from each
  evaluation range: the first 1\,000 blocks of the
  main evaluation (23\,699\,751--23\,700\,750) and
  the full 1\,000-block ArbiNet comparison range
  (24\,100\,000--24\,100\,999).
  These traces let the reviewer run the offline
  binaries, inspect individual transactions, and
  verify that the detection tool produces the same
  verdicts recorded in the evaluation CSV.
  Reproducing the full 220K-block dataset requires
  an Ethereum archive node with
  \texttt{debug\_\allowbreak traceTransaction};
  the online binaries and the export binary support
  this directly.
\end{enumerate}

\noindent
The detection algorithm is an OCaml module inside
a larger open-source blockchain analysis platform
and cannot be cleanly factored out of its dependencies
(decode layer, OSINT enrichment, database layer)
for a self-contained release.
In place of a source bundle, reviewers can validate
every claim through three independent routes: the
Rocq formalization (the algorithm's mechanically
verified specification), the compiled binaries
(which reproduce every verdict in the paper), and
the per-transaction output dataset (which
reproduces every figure and table in
Section~\ref{sec:eval}).

During review, all artifacts are hosted at
\url{https://anonymous.4open.science/r/artifact-submission-2128}.
Once anonymity is lifted, a link to the public
source repository will be added, pinned at the
commit matching the submitted version.

\section{AMM Pricing and Arbitrage Condition}
\label{app:amm}

The detector is invariant-agnostic.  Across the
220\,000-block range, every AMM family we
encountered,
constant-product (Uniswap~v2~\cite{uniswap}),
concentrated liquidity (Uniswap~v3~\cite{uniswapv3}
and~v4), StableSwap (Curve~\cite{curve}), weighted
pools (Balancer~\cite{balancer}), proactive market
making (DODO), and adaptive-fee designs (Algebra,
Trader Joe's Liquidity Book), exposes the same
external interface to the EVM trace: a pair of
standard ERC-20 \texttt{Transfer} events within a
call frame, one for each leg of the swap.
Rule~$R_1$ chains the in-leg with the out-leg;
rule~$R_5$ handles the Uniswap~v4
\texttt{PoolManager}, whose singleton design emits
\texttt{Transfer} events at the manager boundary rather
than per pool, so that a pass-through is chained rather
than mistaken for a two-cycle.
Multi-token pools (Curve $n$-pools, Balancer
8-token pools) are irrelevant at this layer: the
pool's internal $n$-ary structure collapses to a
binary in/out at the trace boundary.
The arbitrage condition, different effective
exchange rates between two pools holding the same
token pair, generalizes uniformly across
invariants: it is the existence of a
profit-preserving cycle in the realized transfer
graph, which our algorithm detects structurally
without ever consulting the underlying pricing
function.

\section{Proofs}\label{app:proofs}

The five theorems of \S\ref{sec:algo:formal} are
mechanized in Rocq~$\geq$\,9.0.  The development is
one self-contained file of $13\,501$ lines with no
external libraries, comprising the five
\texttt{Theorem}s, 432 lemmas, 13 corollaries and
one worked counterexample, and closing with no
\texttt{Admitted} obligation and no \texttt{Axiom};
\texttt{Print Assumptions} on each theorem lists
only the eleven declared \texttt{Parameter}s
described below.  The artifact is the authoritative
reference.\footnote{\url{https://anonymous.4open.science/r/artifact-submission-2128/rocq/Arbitrage.v}}
This appendix maps each paper statement to its
mechanized counterpart, gives compact proofs whose
structure follows the mechanization, and describes
the development itself; the artifact also carries a
long-form companion that expands these
proofs.\footnote{\url{https://anonymous.4open.science/r/artifact-submission-2128/proofs/proofs.pdf}}

\subsection{What Is Guaranteed, and What Is Not}
\label{app:guarantees}

The paper's claims have different standings, and
Table~\ref{tab:guarantees} separates them.  The
distinction that matters most is the last group: the
theorems are about the transfers the decoder produces,
so decoder fidelity, the correspondence between the
production implementation and the mechanized kernel,
and any claim about economic intent are outside what
is proved.

\begin{table}[h]
\centering
\footnotesize
\caption{Standing of each claim.}
\label{tab:guarantees}
\begin{tabular}{@{}>{\raggedright\arraybackslash}p{5.3cm}>{\raggedright\arraybackslash}p{2.7cm}@{}}
\toprule
Claim & Standing \\
\midrule
No rule invents, drops or duplicates a transfer
  & mechanized \\
Reduction terminates under any rule order
  & mechanized \\
$3n-2$ pass bound on fully lifted trees
  & derived from mechanized budget and step decrease \\
Balances agree under every rule order
  & mechanized \\
Unique normal form under the deterministic kernel
  & mechanized \\
Joinable iff equal normal forms; equivalence up to
reassociation decidable
  & mechanized \\
\midrule
Verdict yields a Definition~\ref{def:arb} witness over
$T_0$'s transfers
  & mechanized; hypotheses hold at $T_0$ \\
Local joinability up to reassociation
  & mechanized on well-formed trees \\
\midrule
Decoded transfers faithfully reflect value movement
  & assumed (\S\ref{sec:limits}) \\
Production module agrees with the extracted kernel
  & argued, not measured \\
Promotion path (0.8\% of detections)
  & outside the mechanized core \\
Economic intent behind a detected cycle
  & not addressed \\
Recall
  & measured, not guaranteed \\
\bottomrule
\end{tabular}
\end{table}

\subsection{Theorem Map}\label{app:thmmap}

Table~\ref{tab:thmmap} pairs each numbered result
with the artifact statement that carries it, together
with what the mechanization adds beyond the paper's
phrasing.  The remaining rows are the supporting
results the text names.

\begin{table}[h]
\caption{Paper results and their mechanized
  counterparts.}
\label{tab:thmmap}
\centering
\footnotesize
\setlength{\tabcolsep}{3pt}
\begin{tabular}{@{}l>{\raggedright\arraybackslash}p{4.6cm}@{}}
\toprule
Paper & Artifact statement, and what it adds \\
\midrule
Thm~\ref{thm:preserve}
  & \texttt{theorem\_1\_preservation}: walk
    tracking, transfer inclusion, and multiset
    conservation \\
Thm~\ref{thm:term}
  & \texttt{theorem\_2\_termination}: kernel
    normal form; the $3n-2$ bound
    (\texttt{termination\_bound}); well-foundedness
    under every rule order
    (\texttt{rewrite\_step\_wf}) \\
Thm~\ref{thm:sound}
  & \texttt{theorem\_3\_soundness}: verdict
    \textit{Arbitrage} yields a
    Definition~\ref{def:arb} witness over $T_0$'s
    own transfers (\texttt{def5\_witness}) \\
Thm~\ref{thm:confluence}
  & \texttt{theorem\_4\_confluence}: unique kernel
    normal form.  Clause 2 is
    \texttt{observable\_confluence\_delta}; clause 3
    is \texttt{local\_confluent\_mod\_kappa\_holds} \\
Thm~\ref{thm:decidable}
  & \texttt{theorem\_5\_decidable\_equivalence}:
    joinable iff equal normal forms.  The coarser
    equivalence is \texttt{struct\_equiv\_dec} \\
\midrule
Def.~\ref{def:arb}
  & \texttt{def5\_witness}: the definition in
    transfer vocabulary only; its bundle size is
    pinned by \texttt{walk\_decomposition\_count} \\
no fabrication
  & \texttt{no\_fabricated\_transfers}: a chain's
    transfers are a sub-multiset of $T_0$'s \\
reported profits
  & \texttt{reported\_deltas\_are\_input\_deltas}:
    every reported balance is a balance of the
    input trace \\
plain swap
  & \texttt{plain\_swap\_not\_def5}: a two-leg swap
    fails Definition~\ref{def:arb} on the
    round-trip clause alone \\
\bottomrule
\end{tabular}
\end{table}

\subsection{Preservation}

A chain $C$ \emph{tracks} walks in $G$ when its
transfer sequence decomposes into connected runs,
each a walk over edges of $E$, and the stored
endpoints of $C$ are the source of its first
transfer and the destination of its last.  The
invariant holds at $T_0$, where every leaf is one
edge, and is preserved by every rule.  Chaining
(R1--R6, R10--R12) concatenates two operands whose
junction addresses match, so two walks join into
one at a shared vertex.  Lifting repositions
subtrees and touches no chain.  Merging (R7--R9,
R13) records the union of its operands' walks; the
merged chain's transfer list is the concatenation
of the operands' lists, so its decomposition into
connected runs is the disjoint union of theirs.
Annotation (R14, R15) relabels and moves nothing.
For conservation, observe that every rule's
conclusion carries exactly the transfers of its
operands, concatenated: no rule introduces,
drops, or duplicates a transfer, so the leaf
multiset of any reachable tree is a permutation of
$T_0$'s.  The mechanization proves the three
clauses by one induction over
\texttt{rewrite\_step} and its reflexive-transitive
closure. \hfill$\square$

\subsection{Termination}

Let $u(T)$ count the chains of $T$ not yet
labeled and $c(T)$ the total number of children
across $\mathsf{Tree}$ nodes.  Labels
are monotonic: annotation only adds them, a merge
labels its result at once, and the single
downgrade (R16) runs after the fixpoint.  The two
arguments order the pair differently.  The
annotation-and-connection loop uses
$(u(T), c(T))$.  The unrestricted relation uses
$(c(T), u(T))$: chaining and merging replace two
children by one and lifting removes an internal
node, each decreasing~$c$, while annotation holds
$c$ fixed and decreases~$u$.  Every rule decreases
that measure, which gives well-foundedness under
every rule order
(\texttt{rewrite\_step\_decreases}).  For the
bound, chains arise only from
leaves, so $u(T_0) \leq n$; on a fully lifted,
non-empty tree every internal node has at least
two children, so by handshaking
$c(T_0) \leq 2n - 2$.  Each non-trivial pass
consumes at least one unit of the budget $u + c$,
hence at most $3n - 2$ passes.  What is mechanized
is the budget $u + c \leq 3n - 2$
(\texttt{termination\_bound}, composing
\texttt{unlabeled\_le\_transfers} with
\texttt{cc\_plus2\_le\_twice\_ct}, the handshake on
fully lifted trees) and the loop's progress
(\texttt{fixpoint\_\allowbreak step\_\allowbreak
det\_\allowbreak decreases}); together
with the component monotonicity these yield the pass
bound.  Theorem~\ref{thm:term} states the budget, not
the derivation length directly. \hfill$\square$

\subsection{Soundness}

The proof has the three steps announced in
\S\ref{sec:algo:formal}: an inversion, a bridge,
and a transfer.

\emph{Inversion.}  The cascade of
Algorithm~\ref{alg:classify} returns
\textit{Arbitrage} only when no verdict gate
fired, so the reduced tree holds at least one
chain $C$ labeled \textsc{arb}, no leftovers
remain, and by hypothesis $C$ survived $R_{16}$:
$\delta_C[s(C), \tau_{\mathrm{in}}(C)] > 0$.
The theorem's two structural hypotheses on $T_0$ need
no separate argument: the decoder emits leaves and
call nodes only (Definition~\ref{def:cft}), so $T_0$
carries no chain, and the requirement that its chains
decompose into walks is then vacuous.  Both therefore
hold of every decoded trace, and the theorem's reach
is limited only by the $R_{16}$ premise discussed in
\S\ref{sec:algo:formal}.

\emph{Bridge.}  By preservation, the transfers of
$C$ are transfers of $T_0$, counted with
multiplicity
(\texttt{no\_fabricated\_transfers}).  This is
what makes the remaining steps statements about
the transaction: every property of $C$'s leaves
read off below is a property of edges of $G$.

\emph{Transfer.}  Four facts about $C$ descend to
those edges.  First, closure: R14 fires only on
$s(C) = d(C)$, and the invariant that a chain's
stored endpoints are its first source and last
destination makes the underlying edge sequence
close at $s(C)$.  Second, the token condition: R14
requires
$\tau_{\mathrm{in}}(C) =_\tau
\tau_{\mathrm{out}}(C)$ on the stored boundary
tokens, and the same invariant, for tokens, makes
these the tokens of the first and last edge.
Third, the bundle: the maximal decomposition of
$C$'s edge sequence into connected runs is the
witness Definition~\ref{def:arb} asks for, and its
size is exactly one more than the number of
parallel-merge junctions in $C$
(\texttt{walk\_decomposition\_count}), so
merge-free cycles yield the classical single
closed walk.  Fourth, profit: $\delta_C$ is the
fold of per-transfer balances over $C$'s leaves,
so the surviving $R_{16}$ check is the strict
positivity of Definition~\ref{def:arb}'s third
clause at $v = s(C)$, $\tau =
\tau_{\mathrm{in}}(C)$. \hfill$\square$

Two supporting facts complete the picture.  When
no leftovers remain, the cycles partition the
non-cost transfers, so the aggregate gross balance
is the sum of the per-cycle deltas validated by
$R_{16}$, and the net verdict conditions of
Algorithm~\ref{alg:classify} then bound the
aggregate.  And when
$\tau_{\mathrm{in}} \neq \tau_{\mathrm{out}}$ but
the two are $=_\tau$-equivalent, the economically
correct check sums the delta over the
$=_\tau$-neighborhood of the entry token; this
two-element sum is what the post-fixpoint
promotion path computes, which is why it lives
outside the mechanized core and why
Theorem~\ref{thm:sound} carries its scope
hypothesis.  The converse of the theorem does not
hold, by design: a genuine arbitrage with
leftovers or a mixed balance is reported as
\textit{Warning}, so soundness excludes false
positives and says nothing about false negatives.

\subsection{Uniqueness}

\emph{Clause 1.}  Property~\ref{prop:dse} orders
the children of every node by trace position, so
the kernel's step is a total function: on any
tree it either fires exactly one rewrite, chosen
by a fixed scan order, or reports a normal form.
A deterministic step admits at most one reduction
sequence from any tree; with termination, the
normal form exists and is unique.

\emph{Clause 2.}  Every rule permutes the
transfer multiset (preservation), and the balance
of Definition~\ref{def:balance} is a fold over
that multiset, so any two trees reachable from
$T_0$, under any rule orders, agree on
$\delta[v, \tau]$ at every address and token.
This is a conservation argument, not a confluence
argument, which is why it holds unconditionally.

\emph{Clause 3.}  Freed from the kernel's order,
the relation has genuine critical pairs.  Two
matter.  A bundle of three mutually mergeable
chains can associate as
$(C_1 \| C_2) \| C_3$ or $C_1 \| (C_2 \| C_3)$:
syntactically distinct trees that no further
rewriting reconciles.  And a chain can be at once
an extension operand (R6, R10, R12) and a
parallel-merge operand (R7, R13); firing one rule
destroys the other's redex.  A canonical map
resolves the first family: it flattens every
nest of merges to its operand list and sorts by
trace position, which Property~\ref{prop:dse}
makes a strict order, so both associations have
the same image, and the map is idempotent.  A tree
is \emph{well-formed} when its transfers carry
pairwise-distinct trace positions (linearity, true
of every decoded tree by construction and preserved
by every rule) and no chain that a chaining rule
would consume has a parallel-merge twin or a
pending annotation in its frame.  On such trees,
where the second family's overlap is excluded, any
two one-step reducts of the same tree rejoin up to
this map
(\texttt{local\_confluent\_mod\_kappa\_holds}).
Whether the local statement extends to a global
one is open; every result in the paper uses the
kernel's order and is unaffected.  The companion
treats the critical pairs of the unrestricted
relation, and the exact scope of this theorem and
of Theorem~\ref{thm:decidable}, at
length.\footnote{\url{https://anonymous.4open.science/r/artifact-submission-2128/proofs/proofs.pdf}}
\hfill$\square$

\subsection{Decidability}

If $T \equiv_R T'$, both reduce to a common $U$,
which reduces further to the unique normal forms
of each; uniqueness gives
$T{\downarrow} = T'{\downarrow}$.  Conversely,
equal normal forms witness joinability.
Computing $T{\downarrow}$ terminates
(Theorem~\ref{thm:term}) and tree equality is
decidable, so the word problem is decidable.  For
the coarser equivalence, the canonical map of
clause~3 is computable and idempotent, so
comparing images decides equality up to
reassociation of parallel merges
(\texttt{struct\_equiv\_dec}). \hfill$\square$

\subsection{The Mechanized Development}
\label{app:trust}

The file reads as one narrative in six parts:
the model and the sixteen rules, with R1--R15 as
constructors of one inductive relation and R16 as
the predicate \texttt{validated\_arbitrage}; the
foundational invariants (the measure, preservation,
walk correspondence); the deterministic kernel, a
computable step function mirroring the
implementation's scan order, with its soundness
bridge into the relation; soundness into the
transfer graph, where the chain invariants meet
\texttt{def5\_witness}; the five theorems, stated
together; and the algebra beyond the paper, where
the canonical map, its associativity and
commutativity laws, conservation, and
the decidable structural equivalence live.

The trust base is eleven \texttt{Parameter}s,
declared and never defined; no proof assumes
anything of them beyond their types.  They are the
opaque \texttt{address} and \texttt{token} types
with decidable equalities; the token equivalence
$=_\tau$; the predicates \texttt{is\_burn},
\texttt{is\_mint}, \texttt{is\_singleton\_router}
and \texttt{is\_token\_contract}; the cost model
\texttt{net\_positive}, which makes the theorems
parametric in the deployment's gas accounting (the
attempted tier of \S\ref{sec:eval} is formally the
Definition~\ref{def:arb} cycles on which it
returns false,
\texttt{attempted\_arbitrage\_characterization});
and the trace-order key \texttt{trace\_key}
realizing Property~\ref{prop:dse}'s sibling
order, which a deployment discharges with the
trace's own event indices, unique within a
transaction.  Two well-formedness conditions are stated
and deliberately never assumed, so they are
obligations on a deployment rather than hypotheses
of any theorem: reflexivity and symmetry of
$=_\tau$ (transitivity is not wanted; bridged and
pegged assets do not chain), and injectivity of
the trace key.

The kernel, the verdict cascade, the canonical map
and the decidable equivalence extract to OCaml;
the emitted module compiles and runs against
placeholder realizers, and the extraction
directives (machine integers for the trace key,
native lists and options) are erasure-level and
carry no logical content.  At deployment the
parameters are realized by the decoder's concrete
types, standard event-signature matching for burn
and mint, and the optional router set of R5.

\section{Structural Equivalence in Practice}
\label{sec:equivalence}

\subsection{Structural Equivalence by Example}
\label{sec:equiv:example}

We illustrate structural equivalence on two
confirmed arbitrages from the evaluation dataset;
the companion develops this material, including the
query language and the family grammar, at greater
length.\footnote{\url{https://anonymous.4open.science/r/artifact-submission-2128/proofs/proofs.pdf}}
The two transactions involve entirely different
tokens, different pool contracts, and a different
number of internal swaps, yet their reductions
converge to the same topological shape.

\paragraph{$\mathit{tx}_R$.}
The running example from Section~\ref{sec:algo}
(Figure~\ref{fig:txR}).  A smart contract receives
$0.0575$~WETH and splits it into three parallel
swap paths: one through a \textsc{weth/moon} pool,
one through \textsc{chad/weth}, and one through
\textsc{duck/weth}.  Each path converts its
intermediate token (\textsc{moon}, \textsc{chad},
\textsc{duck}) into \textsc{bean} at a dedicated
second-hop pool.  All three \textsc{bean} streams
converge on a single \textsc{bean/weth} pool, which
returns $0.0598$~WETH to the smart contract.  The
raw trace contains 10~transfer events across 6~pool
interactions and produces 16~reduction steps.

\paragraph{$\mathit{tx}_S$.}%
\footnote{\texttt{0x\seqsplit{aa4b2aec9c7a571f96062576135632ece91bde142075963219b26232404dfc5c}}}%
A different smart contract belongs to the same
structural family with more paths.  $\mathit{tx}_S$ has the
same star topology as $\mathit{tx}_R$
(Figure~\ref{fig:txR}), with six parallel paths
through twelve pools instead of three through six:
WETH is split across pools \textsc{puppies},
\textsc{Terminus}, \textsc{tomoe}, \textsc{yfi},
\textsc{cpool}, and \textsc{zeta}; each path
converts its intermediate token to \textsc{kabosu}
at a dedicated pool, and all six \textsc{kabosu}
streams converge on a single \textsc{kabosu/weth}
pool.  The raw trace contains 19~transfer events
across 12~pool interactions.

\paragraph{Reduction}
Despite their different sizes, both transactions
undergo the same reduction sequence.
First, leaf chaining ($R_1$) pairs the inbound and
outbound transfers within each pool call frame.
After lifting, chain extension ($R_6$) connects
per-pool chains into $N$ multi-hop paths
($N{=}3$ for $\mathit{tx}_R$, $N{=}6$ for
$\mathit{tx}_S$).
Then, merge ($R_7$) combines the $N$ paths that
share the same origin and destination into a single
compound chain, and a final chaining ($R_6$)
appends the return leg.
Annotation ($R_{14}$) recognizes the closed cycle
and labels it as an arbitrage.

The normal forms differ only in arity
(3~vs.~6 sub-chains) and in the concrete token
names.  Abstracting over both,
$\mathit{tx}_R{\downarrow}$ and
$\mathit{tx}_S{\downarrow}$ instantiate the same
parametric shape:
\[
  \textsc{weth}
  \;\xrightarrow{\;N\text{ paths}\;}
  \tau_i
  \;\to\;
  \textsc{target}
  \;\to\;
  \textsc{weth}
\]
Their exact normal forms differ, so comparing
normal forms directly separates the two.  What
they share is the skeleton: applying $\mathsf{sk}$
(\S\ref{sec:algo:implications}) to both and
comparing the results decides their common family,
and \S\ref{sec:querylang} casts the two tests as
queries.  Theorem~\ref{thm:decidable} is what makes
either comparison decidable.  The example is chosen
so the shared shape is visually obvious; the same
mechanical procedure decides it for transactions
whose raw traces hold hundreds of transfers, where
visual assessment is impossible.

The equivalence also induces a natural
\emph{taxonomy}.  At the finest level, two
transactions are equivalent iff their normal forms
are syntactically identical; $\mathit{tx}_R$
($N{=}3$) and $\mathit{tx}_S$ ($N{=}6$) fall in
different classes.  Quotienting over arity and
token identities yields the \emph{parametric
families} in which both are one strategy,
star-convergent arbitrage with variable~$N$; the
next subsection gives the families a finite
syntax.

\subsection{A Context-Free Grammar for Strategy
Families}
\label{sec:equiv:grammar}

Write $\mathsf{Star}(\alpha,L,\alpha)$ for an
R14-annotated cycle, $\mathsf{Swap}(\tau,\tau')$
for a single swap, $\circ$ for sequential and
$\parallel$ for parallel composition.  For each
pair $(\alpha, \tau')$ of entry and convergence
token, define
$G_{\alpha,\tau'} =
 (\{S, L, C, H\}, \Sigma, P_{\alpha,\tau'}, S)$
with productions
\begin{align*}
  S &\;\to\; \mathsf{Star}(\alpha,\; L,\;
             \alpha) \\
  L &\;\to\; C \parallel L \;\mid\; C \\
  C &\;\to\; \mathsf{Swap}(\alpha,\tau_1)
             \circ H_{\tau_1} \\
  H_{\tau} &\;\to\; \mathsf{Swap}(\tau,\tau')
             \;\mid\; \mathsf{Swap}(\tau,\tau_k)
             \circ H_{\tau_k}
\end{align*}
Every arm ends at the convergence token~$\tau'$;
the closing $\tau' \to \alpha$ swap is the cycle's
return, recorded by the $\mathsf{Star}$, not part
of any arm.  The \emph{star-convergent family} is
$\mathcal{L}_{\mathsf{star}} =
\bigcup_{(\alpha,\tau')} L(G_{\alpha,\tau'})$.

\paragraph{Why context-free}
The production
$C \to \mathsf{Swap}(\alpha,\tau_1) \circ
 H_{\tau_1}$ binds the convergence token~$\tau'$
across all branches: every arm in a $\mathsf{Star}$
must end at the same intermediate token before
returning to~$\alpha$.  A bottom-up ranked tree
automaton does not enforce this, since sibling
states are computed independently.  Because the
nonterminals $H_\tau$ are indexed by token, what
we give is a grammar \emph{schema}: any finite
token vocabulary instantiates it to a finite
context-free grammar, in which membership is
decidable.  In practice $\Sigma$ is
finite and every concrete word is bounded in depth
(Theorem~\ref{thm:term}) and width (the gas
limit).  New strategy families require only new
productions in~$S$, not changes to the rewriting
system.

\begin{proposition}[Grammar characterization]\label{prop:grammar}
Let $T{\downarrow}$ be a normal form containing an
arbitrage cycle (Definition~\ref{def:arb}) with
hub address~$h$, entry/exit token~$\alpha$, and
one or more parallel swap branches converging at
a common intermediate token~$\tau'$.  Then
$T{\downarrow} \in L(G_{\alpha,\tau'})$.
Conversely, every word in $L(G_{\alpha,\tau'})$ is
the structural image of such a star-convergent
arbitrage.
\end{proposition}

\noindent
\emph{Sketch.} ($\Rightarrow$) By termination and
uniqueness (Theorems~\ref{thm:term}
and~\ref{thm:confluence}), the chaining and
lifting rules (R1, R6) collapse each branch into
a swap sequence
$\mathsf{Swap}(\alpha,\tau_1) \circ \cdots \circ
\mathsf{Swap}(\tau_{k-1},\tau')$, matching
$C \to \mathsf{Swap}(\alpha,\tau_1) \circ
H_{\tau_1}$ with $H$ unfolded as needed (finite by
termination).  Parallel sibling branches at the
hub compose via $L \to C \parallel L$.  The outer
$\mathsf{Star}(\alpha,L,\alpha)$ records that
branches return to hub~$h$'s account in
token~$\alpha$, the closure witnessed by
\textsc{Annotate-Cycles} (R14).
($\Leftarrow$) Any derivation in $G_{\alpha,\tau'}$
exhibits a $\mathsf{Star}$ node whose branches
route tokens through $\tau'$ and return to
$\alpha$; this is a valid CFT normal form by
construction, R14 annotates it on structure, and
$R_{16}$ confirms
$\delta_C[s(C), \tau_{\mathrm{in}}(C)] > 0$
afterwards. $\hfill\square$

Both running examples reduce to derivations in
$\mathcal{L}_{\mathsf{star}}$, with the same
$\mathsf{Star}(\alpha,\,L,\,\alpha)$ skeleton:
$\mathit{tx}_R{\downarrow} \in
L(G_{\text{weth},\text{bean}})$ with three swap
arms and
$\mathit{tx}_S{\downarrow} \in
L(G_{\text{weth},\text{kabosu}})$ with six.
The grammar generates an infinite family;
$\mathit{tx}_R$ and $\mathit{tx}_S$ are two
concrete members.

\subsection{Extending the Grammar: Flash-Loan
Wrappers}\label{sec:grammar:flashloan}

Flash-loan-wrapped arbitrages, the family that
contains 51\% of our exclusive confirmed
detections (Category~2, \S\ref{sec:eval:manual}),
extend the star grammar without changing the
rewriting system.  The borrow and repay legs of a
flash loan form a specular pair
(\S\ref{sec:algo:leftover}); leftover recovery
removes them from the leaf list and emits a
\emph{leftover-cycle} node $\mathsf{LC}(\alpha, a)$
recording the loan token and amount.  Leftover
cycles are preserved by every later rewrite rule,
so they survive into the normal form alongside
the star-convergent cycle that uses the borrowed
capital.

A normal form is therefore a pair
$(N, \Lambda)$, where $N$ is a star derivation
(\S\ref{sec:equiv:grammar}) and $\Lambda$ is a
(possibly empty) sequence of $\mathsf{LC}(\alpha,
\cdot)$ nodes whose token coincides with the
star's hub.  Adding the production $\Lambda \to
\varepsilon \mid \mathsf{LC}(\alpha, a) \cdot
\Lambda$ to $G_{\alpha,\tau'}$ yields the
extended grammar $G^F_{\alpha,\tau'}$; when
$\Lambda = \varepsilon$ the extension reduces to
the star grammar, and the flash-loan family
$\mathcal{L}_{\mathsf{flash}} \subset \bigcup
L(G^F_{\alpha,\tau'})$ is the subset with
$\Lambda$ non-empty.  Membership remains
decidable in linear time.  Specular pairs on
tokens unrelated to the cycle (refunds,
meta-transaction relays) do not satisfy any
$G^F_{\alpha,\tau'}$ and fall outside the family.

\paragraph{The 17\,912~WETH arbitrage as a
concrete member.}
We illustrate on Ethereum transaction
\texttt{0x\seqsplit{45388b0f9ff46ffe98a3124c22ab1db2b1764ecb3b61234e29e5c9732b7fd4ab}},
a confirmed arbitrage from the exclusive
detections EigenPhi does not flag (full forensic
analysis in Appendix~\ref{sec:casestudy:fp}).
A bot borrows 14\,175.71~WETH from an aggregator,
executes a three-hop cycle through Bancor, a
BNT/AAVE pool, and an AAVE/WETH pool, and repays
the loan.

After reduction, the normal form is the pair
\begin{align*}
S \;=\;& \mathsf{Swap}(\textsc{weth},\textsc{bnt})
  \circ \mathsf{Swap}(\textsc{bnt},\textsc{aave}),\\[-1pt]
(N,\Lambda) \;=\;& \bigl(\mathsf{Star}(\textsc{weth},
  \,S,\,\textsc{weth}),\\[-1pt]
 &\ \ \mathsf{LC}(\textsc{weth},\,14\,175.71)\bigr).
\end{align*}
The
single arm ends at the convergence
token~\textsc{aave}, and the closing
$\textsc{aave} \to \textsc{weth}$ swap is the
cycle's return, recorded by the $\mathsf{Star}$
exactly as $\textsc{kabosu} \to \textsc{weth}$ is
for~$\mathit{tx}_S$.  The word lies in
$L(G^F_{\textsc{weth},\textsc{aave}})$
and hence in $\mathcal{L}_{\mathsf{flash}}$.
The profit unwrap (a separate star on WETH
under the native/wrapped equivalence) does not
participate in the family-membership test.
Membership is decidable in linear time
(parse $(N,\Lambda)$, check leftover-cycle
tokens), and well-defined under the deterministic
reduction order by Theorem~\ref{thm:confluence}.

\subsection{From Detection to Query Language}
\label{sec:querylang}

We write the three queries of
\S\ref{sec:algo:implications} as
$Q_{\mathsf{arb}}$ (detection), $Q_{\mathsf{eq}}$
(exact structural equivalence) and
$Q_{\mathsf{fam}}$ (family equivalence via the
skeleton map $\mathsf{sk}$), add two further
queries useful for forensic and monitoring work,
and demonstrate $Q_{\mathsf{fam}}$ on the running
example.

\paragraph{$Q_{\mathsf{match}}$: strategy
matching.}
Given a reference transaction
$T_{\mathsf{ref}}$,
$Q_{\mathsf{match}}(T) = \top$ iff
$T{\downarrow} = T_{\mathsf{ref}}{\downarrow}$.
Setting $T_{\mathsf{ref}} = \mathit{tx}_R$ flags
every transaction whose normal form exhibits
the same three-arm star, regardless of pool or
protocol identity.

\paragraph{$Q_{\mathsf{anom}}$: structural
anomaly.}
Given a set $S$ of normal forms previously
observed for a bot address,
$Q_{\mathsf{anom}}(T) = \top$ iff
$T{\downarrow} \notin S$.  A bot whose past
behaviour is captured by $S$ but later executes a
transaction with a different shape (e.g.\ a
four-arm star, a sequential chain) lands outside
$S$ and is flagged.  Both queries are decidable
in polynomial time.

\paragraph{Worked example:
$Q_{\mathsf{fam}}(\mathit{tx}_R, \mathit{tx}_S)
= \top$.}
The skeleton map $\mathsf{sk}$ erases path length,
arity, and token identity
(\S\ref{sec:algo:implications}).  On
$\mathit{tx}_R{\downarrow} = \mathsf{Star}(
\textsc{weth},\, C_1 \parallel C_2 \parallel C_3,\,
\textsc{weth})$, where each $C_i$ is a two-hop
chain through \textsc{bean}: erasing path length
yields $\mathsf{Swap}(\textsc{weth},
\textsc{bean})$ per arm; erasing arity collapses
the three identical arms to one; erasing tokens
gives
$\mathsf{Star}(\star, \mathsf{Swap}(\star,\star),
\star)$.
The same computation on
$\mathit{tx}_S{\downarrow}$ (six two-hop arms
through \textsc{kabosu}) yields the same
skeleton.  Hence
$Q_{\mathsf{fam}}(\mathit{tx}_R, \mathit{tx}_S)
= \top$ despite different arities, tokens, and
pools.  A sequential arbitrage (one cycle
feeding into the next) would have a different
skeleton topology and a different family.

\paragraph{Richer queries}
Combining structural predicates with the per-cycle
balances~$\delta_C$ (computed during normalization)
yields decidable queries such as ``arbitrages
with profit $>v$ through at least three pools''.
The number of cycles in~$T{\downarrow}$ is finite
(Theorem~\ref{thm:term}), so $\delta_C$-aware
queries remain decidable.  The rewriting system
thus underpins a domain-specific language for
fund-flow analysis; this paper instantiates one
query, $Q_{\mathsf{arb}}$.

%% ---------------------------------------------------------
\section{Forensic Case Studies}
\label{sec:casestudy}
%% ---------------------------------------------------------

The reduced AST is not only a detection artifact;
it is a readable record of what happened in a
transaction.  We demonstrate this on three
transactions from the evaluation dataset: a
stablecoin-migration exploit whose profit source
only the tree exposes, a confirmed arbitrage that
EigenPhi misses, and a complex DeFi operation that
\emph{looks} cyclic but is not.  In each case the
AST tells the full story.

\subsection{A Stablecoin Migration Exploit}
\label{sec:casestudy:tp}

Transaction
\texttt{0x\seqsplit{4fe7f35c14e51aa31460c03502928d40c036261234dafba99ef4fa45d0b66904}}.

\paragraph{Transaction profile}
The sender and receiver are the same address
(\addr{0xe834a4106adb7a639af2e12567308238d8b1897b}):
a self-calling MEV bot.
The EVM trace contains 84~internal calls reaching
depth~9.  The decoded trace yields 14~ERC-20
Transfer events across 5~tokens: agEUR, EURe
(Monerium EUR emoney), EUROC, a vault share token,
and WETH.

\paragraph{AST reduction}
The algorithm reduces the initial AST in 10~steps.
The final tree contains 10~chain nodes, all with
$s(C) = d(C)$ (closed loops rooted at the bot
address), plus four leftover leaves.
Two of the ten chains satisfy
$\tau_{\mathrm{in}} =_\tau \tau_{\mathrm{out}}$ and are
annotated as arbitrage cycles; the remaining eight
are individual swaps:

\smallskip
\begin{small}
\begin{tabular}{@{}llll@{}}
\toprule
$\tau_{\mathrm{in}}$ & $\tau_{\mathrm{out}}$
  & Label & Role \\
\midrule
agEUR & agEUR & \textit{arbitrage}
  & Main cycle \\
EURe & EURe & \textit{arbitrage}
  & Profit extraction \\
\midrule
agEUR & EURe & chaining & Swap \\
EURe & EUROC & chaining & Swap \\
EUROC & agEUR & chaining & Swap \\
agEUR & WETH & chaining & Residual sale \\
WETH & ETH & \textit{token burn} & WETH unwrap \\
\bottomrule
\end{tabular}
\end{small}
\smallskip

\noindent The three swap chains reveal the arbitrage
route: agEUR enters a pool that returns EURe, the
EURe enters a second pool that returns EUROC, and
the EUROC enters a third pool that returns agEUR.
The fixpoint loop connected these three chains into
the main cycle (first row).

\paragraph{The token-migration twist}
The dominant profit comes from an unexpected source.
One of the pools involved in the route holds
liquidity in two versions of the same stablecoin:
the original EURe contract\linebreak
(\addr{0x3231cb76718cdef2155fc47b5286d82e6eda273f})
and its post-migration successor\linebreak
(\addr{0x39b8b6385416f4ca36a20319f70d28621895279d}).
When the bot swaps agEUR at this pool, it receives
returns denominated in \emph{both} contracts.
It routes the new-contract tokens onward through
the cycle and keeps the old-contract tokens,
netting 2{,}413~EURe in the process.
The algorithm detects this without knowing that
a token migration occurred: it simply observes that
EURe enters and exits the bot at different amounts,
producing a positive delta.

\paragraph{Flash loan}
The bot funds the operation with a 2{,}413.95~agEUR
flash loan from the Balancer vault.  The borrow
and repay transfers form a specular pair (same
token, same amount, reversed addresses), collapsed
into a single leftover cycle
$\mathsf{LC}(\text{agEUR}, 2{,}413.95)$ by the
recovery pass.

\paragraph{Profit}
The delta map at the bot address:

\smallskip
\begin{small}
\begin{tabular}{@{}lrl@{}}
\toprule
& Amount & Token \\
\midrule
Gross $\delta$ & $+2{,}413.39$ & EURe \\
Gross $\delta$ & $+0.000124$ & agEUR \\
Gross $\delta$ & $+0.000636$ & ETH \\
Gas cost & $-0.000182$ & ETH \\
\midrule
\textbf{Net} & $\mathbf{+2{,}413.39}$ & EURe \\
\textbf{Net} & $+0.000124$ & agEUR \\
\textbf{Net} & $+0.000453$ & ETH \\
\bottomrule
\end{tabular}
\end{small}
\smallskip

\paragraph{Verdict}
Confirmed arbitrage.
Two cycles, zero leftovers, positive final balance.
The AST tells the full story: the arbitrage route,
the funding mechanism (Balancer flash loan), and
the profit source (old vs.\ new EURe contract)
are all readable from the reduced tree, without
any knowledge of the contracts involved.

\subsection{A 17{,}912~WETH Arbitrage That EigenPhi Misses}
\label{sec:casestudy:fp}

This is the flash-loan arbitrage introduced in
\S\ref{sec:grammar:flashloan}
(tx~\texttt{0x\seqsplit{45388b0f9ff46ffe98a3124c22ab1db2b1764ecb3b61234e29e5c9732b7fd4ab}}).
Three anchors: $17{,}912$~WETH gross cycle profit
(17.72 in, 17{,}930 out), $14{,}175.71$~WETH
flash-loan borrow/repay, $13{,}088$~ETH builder
bribe.  EigenPhi misses it; our system confirms.
Here we summarise the reduced AST and the profit
accounting.

\paragraph{Transaction profile}
The EOA
(\addr{0x5884b2faa9ad6f38010831e2290e515af17a7d47})
calls a bot contract
(\addr{0x06cff7088619c7178f5e14f0b119458d08d2f5ef}).
The trace is compact: 42~internal calls reaching
depth~8, with only 6~ERC-20 Transfer events across
4~tokens: WETH, BNT (Bancor Network Token), AAVE,
and native ETH.

\paragraph{AST reduction}
The algorithm reduces the AST to two chains
both labelled \textit{arbitrage} and both
WETH\,$\to$\,WETH closed on the bot address with
zero leftovers: a main cycle and a profit unwrap.
The seven legs of the main cycle, in order, are:
the bot wraps 17.72 WETH at the WETH contract,
sends the resulting 17.72 ETH to Bancor and
receives 122{,}365 BNT, forwards the BNT to a
BNT/AAVE pool which returns 128.6 AAVE, and the
AAVE/WETH pool returns 17{,}930 WETH to the bot.

\paragraph{The price discrepancy}
The AST reveals a striking imbalance:
17.72~ETH buys 122{,}365~BNT at Bancor, which
buys 128.6~AAVE, which is worth 17{,}930~WETH
at the AAVE/WETH pool.  This is a
$1{,}000\times$ return through a three-hop path.
The discrepancy is at the AAVE/WETH pool, which
overvalues AAVE relative to Bancor.  Whether this reflects
a temporary liquidity event, a governance action,
or a pool manipulation is beyond the scope of the
algorithm; the AST records the flows and the
profit, leaving the interpretation to the analyst.

\paragraph{Flash loan and builder payment}
The bot borrows 14{,}175.71~WETH from an
aggregator
(\addr{0xbbbbbbbbbb9cc5e90e3b3af64bdaf62c37eeffcb});
the borrow/repay pair reduces to a single
leftover cycle $\mathsf{LC}(\text{WETH},
14{,}175.71)$.  The bot unwraps the remaining
WETH and forwards 73\% of the gross profit to the
block builder as a priority-inclusion bribe.

\paragraph{Verdict}
Confirmed arbitrage: two cycles, zero leftovers,
gross $+17{,}912$~WETH netting to $+4{,}824$~ETH
to the EOA after the $13{,}088$~ETH builder bribe
and gas.  The entire analysis derives from
6~ERC-20 transfers and 8~ETH value transfers,
with no knowledge of Bancor, AAVE, or any of the
pool contracts involved.

\subsection{A Yield Harvest That Looks Like an Arbitrage}
\label{sec:casestudy:tn}

Transaction
\texttt{0x\seqsplit{5d44c672cce36070dae0419d4cd2852f05bac823637e2eb58e2c403847c080d9}}.

EigenPhi flags this transaction as an arbitrage.
Our system does not.  The disagreement does not
rest on the rewriting: the decoded transfers alone
show that no token returns to its origin in the
denomination it left in, and that the sender's net
balance is negative.  Either fact can be checked
against the trace without running the reduction.
The reduced AST is how we present that conclusion
compactly, not how we establish it.

\paragraph{Transaction profile}
The EVM trace is large: 346~internal calls reaching
depth~11, with 24~ERC-20 Transfer events across
11~distinct tokens.  The tokens include WETH, USDC,
CRV, CVX, crvUSD, frxUSD, and five LP or vault
share tokens.  At first glance, the transaction
resembles a complex multi-venue arbitrage.

\paragraph{AST reduction}
The algorithm constructs the initial AST
(31~leaf nodes) and reduces it in 13~steps.
The final tree contains 8~chain nodes, all with
closed loops ($s(C) = d(C)$), plus leftover
transfers and burns.

\paragraph{Closed-loop analysis}
The reduction yields eight closed loops, but
\emph{none} has
$\tau_{\mathrm{in}} =_\tau \tau_{\mathrm{out}}$:
three are LP unwraps from a single Curve LP
contract (entry token \texttt{ecb0}, exits to
\texttt{c0c17dd0}, \texttt{crvUSD}, and a
CVX$\to$WETH reward sale); three are stable swaps
among \texttt{crvUSD}, \texttt{frxUSD},
\texttt{c0c17dd0}, and USDC; and the remaining
two are a USDC$\to$WETH swap and a CRV reward
sale through the WETH/CRV pool.

\noindent Every row is a swap: tokens enter and exit
the same address, but in different denominations,
so no row satisfies R14 or R15.  USDC, which both
enters and exits the central contract, converges
from multiple pools to the sender rather than
returning to its origin: the pattern is a funnel,
not a cycle.

The overall token flow is linear:
LP\,tokens\,$\to$\,ecb0\,$\to$\,\allowbreak
c0c17dd0/crvUSD\,$\to$\,frxUSD\,$\to$\,\allowbreak
crvUSD\,$\to$\,USDC\,$\to$\,EOA,
with side branches for CRV and CVX reward claims
sold for WETH and then for USDC.  This is a
\emph{yield harvesting} operation: the sender
burns LP tokens from a Curve vault, claims
accrued rewards, converts everything to USDC, and
withdraws.  No token completes a round trip, and
the final balance is negative (gas exceeds the
tiny ETH residual), consistent with a position
unwind rather than a profit-seeking trade.  Not an arbitrage: no chain the reduction builds
returns in the asset it left in, and the final
balance is negative.  The algorithm reaches this
conclusion without knowing what Curve, Convex, or
any of the vault contracts are: 24~transfers,
8~closed loops, no matching boundary tokens.
For completeness we also ran
ArbiNet~\cite{arbinet2024}, the GNN baseline of
Section~\ref{sec:eval}, on this transaction's block:
it does not flag the transaction, so the false
positive is specific to EigenPhi here, not to
flow-based detection as a class.

\paragraph{What a flat detector sees, and why}
Projected onto the transfer graph
(Definition~\ref{def:tg}), this transaction does
contain address-level cycles: three tokens both enter
and leave the central contract.  On this graph,
EigenPhi's documented pipeline~\cite{eigenphi_method}
fires mechanically: every pool swap sends one token out
and receives another back, so the contract and its
pools are mutually reachable and collapse into one
strongly connected component, and the sender's
$19{,}212$~USDC satisfies the positivity test.  Yet the
component is not a traversable trading route: no
sequence of transfers carries value round, because the
token handed on at each junction is not the token
received.  What the aggregate view shows is only that
the same asset came in and went out, in these amounts:

\begin{center}\footnotesize
\begin{tabular}{@{}lrrr@{}}
\toprule
round trip & out & in & balance\\
\midrule
WETH      & $0.0492$      & $0.0492$      & $0$\\
LP token  & $23{,}048.44$ & $23{,}048.44$ & $0$\\
crvUSD    & $30{,}719.00$ & $30{,}716.40$ & $-2.60$\\
\bottomrule
\end{tabular}
\end{center}

\noindent
The WETH row is the one that most resembles an
arbitrage.  It is real in the aggregate: $0.0492$~WETH
leaves the contract for a WETH/USDC pool, and $0.0488$
and $0.0004$~WETH come back from the pools where the
position's CRV and CVX rewards were sold.  The sum is
exact, so the contract gains nothing.

Neither does any of it compose into a chain.  The
outgoing transfer and the two returning ones were
emitted by three different calls, at depths $9$, $8$
and $7$ of the call tree, so they are never siblings
and no chaining rule can join them.  Even setting the
call structure aside, the tokens do not line up: the
contract receives USDC from the pool it sent WETH to,
and what it sends next is CRV, so there is no
continuity to chain on.  The reduction therefore does
not build a cycle here and reject it on its balance;
it builds no such cycle at all.  The sender settles
the question on its own: it holds a single edge in the
whole graph, receiving $19{,}212$~USDC and sending
nothing, so it lies on no cycle.

\section{Cross-Chain Portability}
\label{sec:arbitrum}
%% ---------------------------------------------------------

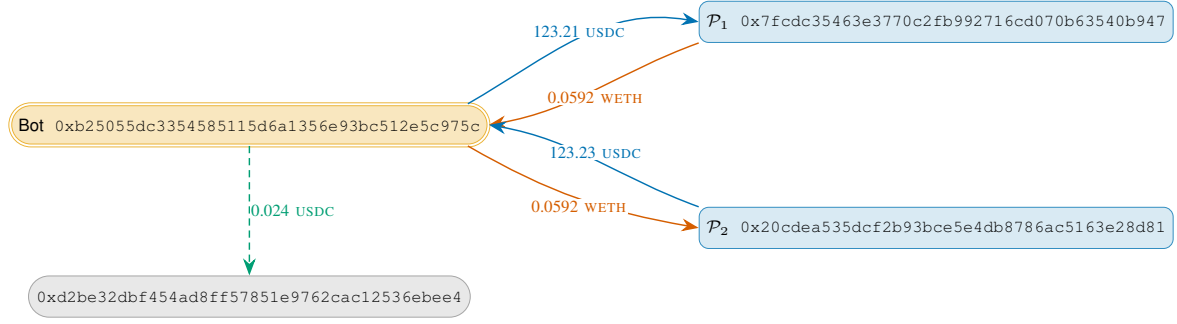
\begin{figure*}[t]
\centering
\resizebox{0.85\textwidth}{!}{%
\begin{tikzpicture}[font=\scriptsize,
  acct/.style={draw=gray!70, fill=acctgray!30,
    rounded rectangle, minimum height=6mm,
    font=\scriptsize, inner sep=3pt,
    align=center},
  sc/.style={draw=scamber!80, fill=scamber!25,
    rounded rectangle, minimum height=6mm,
    font=\scriptsize, inner sep=3pt,
    align=center,
    double, double distance=0.8pt},
  pool/.style={draw=poolblue!80, fill=poolblue!15,
    rounded corners, minimum height=6mm,
    font=\scriptsize, inner sep=3pt},
  lbl/.style={font=\scriptsize, fill=white,
    inner sep=0.5pt},
  arr/.style={-{Stealth[length=2.5mm]}, semithick},
]
% Bot (center left)
\node[sc] (bot) at (0,2)
  {\textsf{Bot}\;
  \texttt{0xb25055dc3354585115d6a1356e93bc512e5c975c}};

% Pool 1 (top right)
\node[pool] (p1) at (10,3.5)
  {$\mathcal{P}_1$\;
  \texttt{0x7fcdc35463e3770c2fb992716cd070b63540b947}};

% Pool 2 (bottom right)
\node[pool] (p2) at (10,0.5)
  {$\mathcal{P}_2$\;
  \texttt{0x20cdea535dcf2b93bce5e4db8786ac5163e28d81}};

% Profit receiver (below bot)
\node[acct] (recv) at (0,-0.5)
  {\texttt{0xd2be32dbf454ad8ff57851e9762cac12536ebee4}};

% 1. Bot -> P1 (USDC, blue, upper path)
\draw[arr, okBlue]
  (bot.north east) to[out=30, in=170]
  node[lbl, above, pos=0.5]
  {123.21 \textsc{usdc}} (p1.west);

% 2. P1 -> Bot (WETH, vermillion, return path)
\draw[arr, okVermillion]
  (p1.south west) to[out=200, in=10]
  node[lbl, below, pos=0.5]
  {0.0592 \textsc{weth}} (bot.east);

% 3. Bot -> P2 (WETH, vermillion, lower path)
\draw[arr, okVermillion]
  (bot.south east) to[out=-30, in=170]
  node[lbl, below, pos=0.5]
  {0.0592 \textsc{weth}} (p2.west);

% 4. P2 -> Bot (USDC, blue, return path)
\draw[arr, okBlue]
  (p2.north west) to[out=160, in=-10]
  node[lbl, above, pos=0.5]
  {123.23 \textsc{usdc}} (bot.east);

% 5. Surplus: Bot -> external address (dashed)
\draw[arr, densely dashed, okGreen]
  (bot.south) --
  node[lbl, right, pos=0.5]
  {0.024 \textsc{usdc}} (recv.north);

\end{tikzpicture}%
}
\caption{Transfer graph of an Arbitrum arbitrage
  (tx~\texttt{0x1472c663}).
  \textcolor{okBlue}{Blue}: \textsc{usdc}
  (\texttt{0xaf88d065e77c});
  \textcolor{okVermillion}{vermillion}: \textsc{weth}
  (\texttt{0x82aF4944});
  dashed: surplus forwarded to an external address.
  Detected with no code changes; only the
  \textsc{weth} address was configured.}
\label{fig:arbitrum}
\end{figure*}

The rewriting system depends only on the
deterministic sequential execution property shared
by all EVM-compatible chains
(Property~\ref{prop:dse}), not on Ethereum-specific
assumptions.  As a portability smoke-test of the
model (not a cross-chain evaluation), we ran the
detection algorithm on two chains it was never
designed or tested against: Arbitrum One (an
Ethereum Layer-2 rollup) and BNB Smart Chain (an
independent EVM-compatible Layer-1).

\subsection{Arbitrum}

\paragraph{Setup.}
We selected 1\,000~consecutive Arbitrum blocks
(447\,413\,490--447\,414\,489) containing
23\,250~transactions.
The detection binary was compiled from the same
source code used in the main evaluation
(\S\ref{sec:eval}), with a single
configuration change: the wrapped-asset address was
set to the Arbitrum \textsc{weth} contract
(\texttt{0x82aF\allowbreak 4944\allowbreak 7D8a})
via an environment variable.  No rewriting rule, no
threshold, and no annotation predicate was
modified.  The algorithm received raw transaction
traces from an Arbitrum RPC node and processed them
through the same decode--analyse pipeline used for
Ethereum mainnet.

\paragraph{Results.}
Of the 23\,250~transactions, the system classified
50~as confirmed arbitrages and 952~as warnings
(4.3\% flagged overall).  Confirmed arbitrages
average 2.0~cycles per transaction.  The lower
detection rate compared to Ethereum mainnet reflects
Arbitrum's different transaction mix: most L2
activity consists of user-initiated swaps and
transfers rather than MEV bot strategies.

\paragraph{Case study.}
We examine one representative
transaction in detail.\footnote{%
\texttt{0x\seqsplit{1472c66337ca2caca95b7e5c619cd85c4caa5c7d29292b2dc8966b8bd3ed6d42}}}
The bot contract (\texttt{0xb250}) sends
123.21~\textsc{usdc} to
pool~$\mathcal{P}_1$ (\texttt{0x7fcd}), which
returns 0.0592~\textsc{weth}.  The bot forwards
the entire \textsc{weth} amount to
pool~$\mathcal{P}_2$ (\texttt{0x20cd}), which
returns 123.23~\textsc{usdc}
(Figure~\ref{fig:arbitrum}).
The \textsc{usdc} surplus (0.024~\textsc{usdc}) is
forwarded to a separate address, which may be a
profit receiver, a vault, or part of a larger
multi-transaction strategy.
The system classifies this transaction as
\textit{Warning} rather than \textit{Arbitrage}
because the bot's final balance is mixed: it is net
positive in \textsc{usdc} but net negative in ETH
(consumed as gas).  Without a price oracle, the
algorithm cannot determine whether the overall
position is profitable, so it conservatively issues
a warning.

\paragraph{Structural equivalence.}
After rewriting, the reduced form of this Arbitrum
transaction is structurally identical to the
triangular arbitrages in the Ethereum evaluation:
a single annotated cycle with two chained swaps
and a leftover extraction.  The same rules fired
in the same order and produced the same canonical
shape, despite differences in gas pricing, block
timing, and pool contract implementations between
the two chains.

\subsection{BNB Smart Chain}

\paragraph{Setup.}
To test portability to an independent Layer-1,
we ran the same binary on 1\,000 consecutive BNB
Smart Chain (BSC) blocks
(91\,308\,000--91\,308\,999).
The only configuration change was the wrapped-asset
address, set to the BSC \textsc{wbnb} contract
(\texttt{0xbb4C\allowbreak dB9C\allowbreak bd36})
via the same environment variable.
BSC carries 3-second block times and a mix of
legacy (type~0x0) and EIP-1559 (type~0x2)
transactions in the same block; the decode layer
handles EIP-1559 natively and treats legacy
transactions as outside its current scope.  Legacy
transactions are not malformed: they simply use a
pre-1559 fee field layout the decoder does not
unpack.  Extending coverage would be a decoder
patch, not an algorithmic change; we keep them out
of the sample so the figures below are reported
against a single transaction shape.

\paragraph{Results.}
The 1\,000~blocks contain approximately 105\,000
transactions.  Of these, 38\,126 are EIP-1559
(36\%) and the remaining 66\,874 are legacy
(64\%).  The detection rate is reported against
the EIP-1559 denominator: 88~confirmed arbitrages
and 922~warnings out of 38\,126 EIP-1559
transactions (2.6\% flagged overall).  Against the
full 105\,000 the rate is mechanically lower
(0.96\%), but that figure conflates a covered
sub-population with an out-of-scope one and is
not a meaningful comparison with the Ethereum and
Arbitrum baselines, which run on the post-merge
single-format era.
Confirmed arbitrages average 1.2~cycles per
transaction, with a maximum of~3.
Warnings split into attempted arbitrages
(negative profit after gas, 68\%) and real cycles
downgraded by leftover transfers (32\%),
matching the distribution observed on Ethereum
and Arbitrum.

\paragraph{Case study.}
We examine one confirmed
arbitrage in detail.\footnote{%
\texttt{0x\seqsplit{83c2295b1c17312a0598bb6dc6b9f8f12646f48b30a0423397c8acd7e2504d26}}}
The bot contract (\texttt{0xaf70}) sends
\textsc{wbnb} to pool~$\mathcal{P}_1$
(\texttt{0xe3c1}), which returns
token~\texttt{0xf310}.  The bot forwards it to
pool~$\mathcal{P}_2$ (\texttt{0x3cd4}), which
returns token~\texttt{0x570a}.  Finally,
pool~$\mathcal{P}_3$ (\texttt{0xbffe}) converts
back to \textsc{wbnb} with a positive surplus.
The fixpoint labels the composed cycle as an
arbitrage (R14) with a positive \textsc{wbnb}
balance and zero leftovers:
a clean confirmed arbitrage.

\paragraph{Structural equivalence.}
After rewriting, the reduced form of this BSC
transaction is structurally identical to the
triangular arbitrages on Ethereum and Arbitrum:
a single annotated cycle with three chained swaps,
each through a different pool.  The same rules
(R1 for leaf chaining, R14 for annotation) fire
in the same order and produce the same canonical
shape, despite differences in block times,
gas pricing, and pool contract implementations
across the three chains.

\subsection{Implications}

The two case studies confirm three properties of
the approach.
First, the rewriting rules, the sender
enrichment~$\sigma$, and the cycle-annotation logic
transfer without modification to chains with
different consensus mechanisms (Arbitrum's
centralized sequencer vs.\ BSC's Proof of Staked
Authority), block times (sub-second vs.\ 3\,s),
and gas pricing models.  The only chain-specific
parameter is the wrapped-asset address, a
deployment constant.
Second, OSINT labels (pool names, protocol tags)
are unavailable on both Arbitrum and BSC, yet
detection operates correctly without them.  This
confirms the architectural separation between the
decode and analysis layers: the algorithm reasons
over token flows and sender fields, not over
external metadata.
Third, the detected arbitrages on all three chains
exhibit the same structural patterns (triangular,
multi-hop, batched), produced by the same fixpoint
rules.  MEV strategies are chain-agnostic in
structure; only the token addresses and pool
contracts differ.

\paragraph{Beyond the EVM.}
The three validations cover EVM-compatible chains,
where the trace format is identical by construction.
Property~\ref{prop:dse} also holds for non-EVM
virtual machines.  Porting
to Solana's SVM, for instance, would require only a
new decode layer: SPL~Token instructions (not
events) provide the transfer tuples, and
Cross-Program Invocations form a call tree
analogous to nested \texttt{CALL} frames.  The
$\sigma$ field maps to the CPI invoker.  The
rewriting rules, the fixpoint, and all formal
properties would remain unchanged.

%% ---------------------------------------------------------
\section{ArbiNet Comparison: Case Study}
\label{sec:arbinet:case}
%% ---------------------------------------------------------

To understand the qualitative differences between
the structural and GNN-based approaches, we
examine a randomly chosen block 24\,100\,069 in detail.\footnote{%
Full transaction hashes are recoverable from any
Ethereum explorer by matching the short prefix
within the block.}
Of its 311~transactions, ArbiNet flags 5 as
arbitrages and our system flags 5 as arbitrage
and 1 as warning
(Table~\ref{tab:block69}).  The results are
manually checked.

\begin{table}[t]
\caption{Per-transaction verdicts.
  A~=~ArbiNet, O~=~Ours, E~=~EigenPhi.
  TP~=~true positive (confirmed arbitrage);
  FP\textsubscript{A}/FN\textsubscript{A}~=~ArbiNet
  false positive/negative.  Assessment based on
  manual AST inspection.}
\label{tab:block69}
\centering
\footnotesize
\begin{tabular}{@{}lcccl@{}}
\toprule
Tx & A & O & E & Assessment \\
\midrule
\texttt{0x4573} & arb & arb & arb & TP, 3/3 agree \\
\texttt{0x12f4} & arb & warn & --- & TP, mixed final balance \\
\texttt{0xf170} & arb & arb & arb & TP, 3/3 agree \\
\texttt{0x8dc9} & arb & --- & --- & FP\textsubscript{A}: CoW, $\tau_{\mathrm{in}} \neq \tau_{\mathrm{out}}$ \\
\texttt{0x0fdc} & arb & arb & arb & TP, 3/3 agree \\
\texttt{0x4744} & --- & arb & arb & FN\textsubscript{A}: 3-pool, post-2022 \\
\texttt{0x9d71} & --- & arb & arb & FN\textsubscript{A}: V4, post-2022 \\
\bottomrule
\end{tabular}
\end{table}

\paragraph{Three-way agreement.}
Transactions \texttt{0x4573}, \texttt{0xf170}, and
\texttt{0x0fdc} are flagged as arbitrage by all
three systems.  These form a high-confidence
consensus set: three independent approaches
(protocol-specific heuristics, GNN classification,
and structural rewriting) agree on the same
transactions.

\paragraph{Partial agreement.}
Transaction \texttt{0x12f4} is classified as
\textit{arbitrage} by ArbiNet and as
\textit{warning} by our system; EigenPhi does
not flag it.  The
cyclic structure exists and the cycle is
correctly identified, but gas costs exceed the
gross profit, producing a mixed final balance.
This is a correct verdict under the graduated
confidence model
(\S\ref{sec:classify}): the cycle is real,
but the final balance is mixed (positive in one
token, negative in another after gas), making the
economic outcome indeterminate without a price
oracle.  EigenPhi's absence is consistent: the
transaction's profitability is unclear, and
EigenPhi may require a definitive positive
balance.

We now turn to the cases where the systems
disagree fundamentally.

\begin{table}[t]
\caption{\texttt{0x8dc9}: ArbiNet false positive.
  CoW Protocol batch settlement: both cycles have
  $\tau_{\mathrm{in}} \neq \tau_{\mathrm{out}}$
  (USDC $\neq$ USDT).}
\label{tab:fp8dc9}
\centering
\begin{small}
\begin{tabular}{@{}clll@{}}
\toprule
& From $\to$ To & Amount & Token \\
\midrule
\multicolumn{4}{@{}l}{\textit{Cycle~1
  \textcolor{okPurple}{(cycle, not promoted)}}}\\
1 & \textcolor{okVermillion}{\addr{0xe6eb}
  $\to$ CoW Settlement}
  & 1\,022.06 & \textsc{usdc} \\
2 & \textcolor{okGreen}{CoW Settlement
  $\to$ \addr{0xe6eb}}
  & 1\,022.47 & \textsc{usdt} \\
\addlinespace
\multicolumn{4}{@{}l}{\textit{Cycle~2
  \textcolor{okPurple}{(composite, not promoted)}}}\\
\multicolumn{4}{@{}l}{\footnotesize\textit{%
  swap 1: \textsc{usdc} $\to$ \textsc{bold}
  (Balancer)}}\\
1 & \textcolor{okVermillion}{CoW Settlement
  $\to$ Balancer vault}
  & 1\,021.88 & \textsc{usdc} \\
2 & Balancer vault $\to$ CoW Settlement
  & 1\,022.84 & \textsc{bold} \\
\cdashline{1-4}
\multicolumn{4}{@{}l}{\footnotesize\textit{%
  swap 2: \textsc{bold} $\to$ \textsc{usdt}
  (Uniswap V4)}}\\
3 & CoW Settlement $\to$ Uniswap V4
  & 1\,022.84 & \textsc{bold} \\
4 & \textcolor{okGreen}{Uniswap V4
  $\to$ CoW Settlement}
  & 1\,022.50 & \textsc{usdt} \\
\bottomrule
\end{tabular}
\end{small}
\end{table}

\begin{table}[t]
\caption{\texttt{0x4744}: ArbiNet false negative.
  3-pool multi-hop arbitrage:
  $\tau_{\mathrm{in}}$ = WETH $=_\tau$
  $\tau_{\mathrm{out}}$ = ETH.
  Profit: $+2.0 \times 10^{-5}$ WETH.}
\label{tab:fn4744}
\centering
\begin{small}
\begin{tabular}{@{}clll@{}}
\toprule
& From $\to$ To & Amount & Token \\
\midrule
\multicolumn{4}{@{}l}{\textit{Cycle~0
  \textcolor{okOrange}{(arbitrage,
  fixpoint)}}}\\
\multicolumn{4}{@{}l}{\footnotesize\textit{%
  swap 1: \textsc{weth} $\to$ \textsc{mkr}
  (Balancer)}}\\
1 & \textcolor{okVermillion}{Bot $\to$ Balancer}
  & 0.0105 & \textsc{weth} \\
2 & Balancer $\to$ Bot & 0.0205 & \textsc{mkr} \\
\cdashline{1-4}
\multicolumn{4}{@{}l}{\footnotesize\textit{%
  swap 2: \textsc{mkr} $\to$ \textsc{bnt}
  (Bancor)}}\\
3 & Bot $\to$ Bancor & 0.0205 & \textsc{mkr} \\
4 & Bancor $\to$ Bot & 77.47 & \textsc{bnt} \\
\cdashline{1-4}
\multicolumn{4}{@{}l}{\footnotesize\textit{%
  swap 3: \textsc{bnt} $\to$ \textsc{eth}
  (Pool$_3$)}}\\
5 & Bot $\to$ Pool$_3$ & 77.47 & \textsc{bnt} \\
6 & \textcolor{okGreen}{Pool$_3$ $\to$ Bot}
  & 0.0105 & \textsc{eth} \\
\addlinespace
\multicolumn{4}{@{}l}{\textit{Cycle~1
  \textcolor{okSkyBlue}{(token mint,
  not promoted)}}}\\
1 & Bot $\to$ WETH & 0.01049 & \textsc{eth} \\
2 & WETH $\to$ Bot & 0.01049 & \textsc{weth} \\
\bottomrule
\end{tabular}
\end{small}
\end{table}

\begin{table}[t]
\caption{\texttt{0x9d71}: ArbiNet false negative.
  V4 triangular arbitrage:
  $\tau_{\mathrm{in}}$ = ETH $=_\tau$
  $\tau_{\mathrm{out}}$ = WETH.
  Profit: $+7.2 \times 10^{-6}$ WETH.}
\label{tab:fn9d71}
\centering
\begin{small}
\begin{tabular}{@{}clll@{}}
\toprule
& From $\to$ To & Amount & Token \\
\midrule
\multicolumn{4}{@{}l}{\textit{Cycle~0
  \textcolor{okOrange}{(arbitrage,
  fixpoint)}}}\\
\multicolumn{4}{@{}l}{\footnotesize\textit{%
  swap 1: \textsc{eth} $\to$ \texttt{0xaf04}
  (V4)}}\\
1 & \textcolor{okVermillion}{Bot $\to$ V4}
  & 0.00666 & \textsc{eth} \\
2 & V4 $\to$ Pool & 596.56 & \texttt{0xaf04} \\
\cdashline{1-4}
\multicolumn{4}{@{}l}{\footnotesize\textit{%
  swap 2: \texttt{0xaf04} $\to$ \textsc{weth}
  (Pool)}}\\
3 & \textcolor{okGreen}{Pool $\to$ Bot}
  & 0.00666 & \textsc{weth} \\
\addlinespace
\multicolumn{4}{@{}l}{\textit{Cycle~1
  \textcolor{okBlue}{(token burn,
  not promoted)}}}\\
1 & Bot $\to$ WETH & 0.00666 & \textsc{weth} \\
2 & WETH $\to$ Bot & 0.00666 & \textsc{eth} \\
\bottomrule
\end{tabular}
\end{small}
\end{table}

\paragraph{ArbiNet false positive (\texttt{0x8dc9}).}
This transaction is a CoW Protocol batch settlement
routing through Balancer and Uniswap~V4.
ArbiNet classifies it as an arbitrage; our system
disagrees.  Table~\ref{tab:fp8dc9} shows the
reduced AST: two closed-loop chains, both labeled
\textit{cycle}.  In Cycle~1, a user sends
USDC and receives USDT; in Cycle~2, the settlement
routes USDC through an intermediate token
(\textsc{bold}) and returns USDT.  Both cycles
have $\tau_{\mathrm{in}} = \text{USDC} \neq
\tau_{\mathrm{out}} = \text{USDT}$: the token
equivalence $=_\tau$ identifies only the
native/wrapped pair (ETH/WETH), not stablecoins.
Our system correctly reports zero arbitrage cycles.
One possible explanation is that ArbiNet's GNN
treats stablecoins as interchangeable: both are
pegged to \$1, so a USDC$\to$USDT flow looks
economically equivalent to a same-token cycle.
But structurally it is not: the arbitrage condition
(Definition~\ref{def:arb}) requires token-equivalent
boundaries, not price equivalence.

\paragraph{ArbiNet false negative: \texttt{0x4744}
(multi-hop, 3 pools).}
Both our system and EigenPhi confirm this as an
arbitrage; ArbiNet misses it.
Table~\ref{tab:fn4744} shows the reduced AST:
the bot routes WETH through Balancer
($\to$~MKR), Bancor ($\to$~BNT), and a third pool
($\to$~ETH).  The fixpoint labels the composed
cycle as arbitrage:
$\tau_{\mathrm{in}}$ = WETH $=_\tau$
$\tau_{\mathrm{out}}$ = ETH, $s(C) = d(C)$ = Bot.
A second cycle (token mint) wraps the
profit from ETH to WETH.
ArbiNet misses this transaction because the routing
contracts are absent from its 2022 training data.

\paragraph{ArbiNet false negative: \texttt{0x9d71}
(Uniswap~V4, triangular).}
The pattern repeats with a different topology.
Table~\ref{tab:fn9d71} shows the reduced AST:
the bot sends ETH through the Uniswap~V4
singleton router ($\to$~\texttt{0xaf04}), then a
second pool converts back to WETH.  The fixpoint
labels the cycle as arbitrage:
$\tau_{\mathrm{in}}$ = ETH $=_\tau$
$\tau_{\mathrm{out}}$ = WETH.
A second cycle (token burn) unwraps the
settlement.  Again, ArbiNet misses it because the
V4 singleton postdates its training data.

\paragraph{Emergent swap identification.}
ERC-20 mandates \texttt{Transfer} events but not
swap-specific events such as \texttt{Swap} or
\texttt{TokenExchange}.  Any pair of
\texttt{Transfer} events with complementary
directions and the same~$\sigma$ is structurally
a swap, regardless of declared events.  For a
Curve/Uniswap-V3 transaction the algorithm
recovers 10~swaps from transfers alone, where the
contracts emit only 4 \texttt{Swap} events,
revealing pool interactions invisible to
event-based decoding.  Swap structure is not an
input to the algorithm; it is an output.

\paragraph{Summary.}
The case study illustrates three structural
advantages of the rewriting approach over
GNN-based classification.
First, the token-identity condition at cycle
boundaries filters false positives that
graph-topology classifiers cannot distinguish.
Second, the fixpoint produces formally guaranteed
detections on contracts absent from any training
set, eliminating temporal degradation.
Third, the reduced AST is a readable decomposition
of fund flow into constituent swaps, providing the
analyst with an explainable structural account.

\end{document}